\documentclass[11pt,english]{article}
\usepackage[T1]{fontenc}
\usepackage[latin9]{inputenc}
\usepackage{geometry}
\usepackage{amsmath}
\usepackage{amsthm}
\usepackage{amssymb}
\usepackage{graphicx}
\usepackage{setspace}
\usepackage{caption}
\usepackage{subcaption}
\usepackage{xcolor}
\usepackage{tikz-cd}

\makeatletter

\providecommand{\tabularnewline}{\\}

\theoremstyle{plain}
\newtheorem{thm}{\protect\theoremname}
\theoremstyle{definition}
\newtheorem{defn}[thm]{\protect\definitionname}
\theoremstyle{plain}
\newtheorem{prop}[thm]{\protect\propositionname}

\usepackage{harvard}
\usepackage{booktabs}
\usetikzlibrary{arrows.meta,3d}
\usepackage{tikz-3dplot} 

\makeatother

\usepackage{babel}
\providecommand{\definitionname}{Definition}
\providecommand{\propositionname}{Proposition}
\providecommand{\theoremname}{Theorem}

\usepackage{hyperref}
\begin{document}
\title{Equilibrium Multiplicity:\\
A Systematic Approach using Homotopies,\\ with an Application to Chicago\thanks{I thank Elie Tamer and Holger Sieg for useful comments on this paper's
approach during early stages of the project. }}
\author{Amine Ouazad\thanks{Associate Professor of Finance and Economics, Rutgers Business School and Associate Professor of Economics, HEC Montreal.}}

\maketitle
Discrete choice models with social interactions or spillovers may exhibit multiple equilibria. This paper provides a systematic approach to enumerating them for a quantitative spatial model with discrete locations, social interactions, and elastic housing supply. The approach relies on two \emph{homotopies}. A homotopy is a smooth function that transforms the solutions of a simpler city where solutions are known, to a city with heterogeneous locations and finite supply elasticity. The first homotopy is that, in the set of cities with perfectly elastic floor surface supply, an economy with heterogeneous locations is homotopic to an economy with homogeneous locations, whose solutions can be comprehensively enumerated. Such an economy is $\varepsilon$ close to an economy whose equilibria are the zeros of a system of polynomials. This is a well-studied area of mathematics where the enumeration of equilibria can be guaranteed. The second homotopy is that a city with perfectly elastic housing supply is homotopic to a city with an arbitrary supply elasticity. In a small number of cases, the path may bifurcate and a single path yields two or more equilibria. By running the method on thousands of cities, we obtain a large number of equilibria. Each equilibrium has different population distributions. We provide a method that is computationally feasible for economies with a large number of locations choices, with an empirical application to the City of Chicago. There exist multiple ``counterfactual Chicagos'' consistent with the estimated parameters. Population distribution, prices, and welfare are not uniquely pinned down by amenities. The paper's method can be applied to models in trade and IO. Further applications of algebraic geometry are suggested.

\clearpage{}

\pagebreak{}

\section{Introduction}

Equilibrium multiplicity is a key issue in several fields of economics: 
in dynamic games \cite{tamer2003incomplete}, 
in industrial organization \cite{ciliberto2009market}, in macroeconomics \cite{angeletos2006crises}, 
in the econometrics of games \cite{tamer2010partial,de2013econometric,otsu2023equilibrium}, 
in general equilibrium theory \cite{ghiglino1997multiplicity}. In urban economics and trade, 
the issue of equilibrium multiplicity is particularly difficult as: (i)~economies feature a large
number of discrete choices, such as the number of counties in the United States or the number of tracts or blocks in
a metropolitan area; (ii)~price effects lead to strategic substitutability in choices while social
interactions typically lead to strategic complementarities. These challenges lead to the importance
of propositions and theorems guaranteeing the existence of a single equilibrium in recent literature \cite{allen2020persistence,kleinman2023linear,kleinman2023dynamic}. 
Important empirical contributions have highlighted the presence of multiple equilibria \cite{card2008tipping}
in neighborhood choice, predicted by the classic contributions of \citeasnoun{schelling1969models} and \citeasnoun{benabou1996equity} 
in the two-neighborhood case. Modeling the structure of cities with multiple locations when multiple equilibria
are allowed remains challenging. 

This paper provides a novel approach to the enumeration of multiple equilibria when agents make discrete choices. Discrete choice models, either static or dynamic, are typically used to predict the distribution of populations in geographic areas where agents choose locations based on their amenities, social demographics, and prices. A central concept of this paper is that of \emph{homotopy}, whereby the equilibrium of simpler system, e.g. a city with social interactions where locations have homogeneous amenities and perfectly elastic housing supply, can be differentially ``transformed'', in a way made precise in the paper, into a general equilibrium system, such as those equilibrium equations for a city with social interactions across locations with heterogeneous amenities and finite supply elasticity. A point in this path is an equilibrium vector of demographics for each location and an equilibrium vector of prices. For each infinitesimal shift in amenities (first homotopy) and supply elasticites (second homotopy), a new vector of equilibrium demographics and prices is found.

The end point of the first homotopy is a city with social interactions across locations with heterogeneous amenities and perfectly elastic housing supply. In this first homotopy, all models exhibit strategic complementarities. The starting point of the first homotopy is a city with social interactions across locations with homogeneous amenities and perfectly elastic housing supply. This homotopy is identical to a classical homotopy used to solve systems of polynomial equations with integer degree. A key insight is that, for any arbitrary choice of an $\varepsilon$, the equilibrium equations of cities with social interactions with heterogeneous locations and perfectly elastic housing supply, dynamic or static, McFadden or Fr\'echet, are at $\varepsilon$ distance of a polynomial system of equations with integer degree. This polynomial system can be differentially transformed to a polynomial system whose equations are those of a city with homogeneous amenities and perfectly elastic housing supply. As the question becomes the enumeration of solutions of polynomial systems, this opens the analysis to the mathematical field of algebraic geometry, which studies the zeros of polynomial systems of equations.  Algebraic geometry has a long history in mathematics \cite{dieudonne1972historical}, dating back to the analysis of intersections of curves such as conical sections, starting in 400 B.C.  It is a natural point in the history of urban economics to build on the achievements of the field in connecting quantitative spatial models to linear algebra \cite{kleinman2023linear} and connect the field to the contributions of  algebraic geometry. 

The end point of the second homotopy is a city with social interactions and finite supply elasticity where prices respond to the demand for locations. In these models, there are both strategic complementarities (social interactions) and strategic substitutabilities (price responses). Strategic substitutability may lead to fewer equilibria and/or to equilibria with more similar population distributions. When the economy exhibits both social interactions and price effects (pecuniary externalities), the equilibrium system is not a system of polynomial equations, but differential homotopies are possible -- and formalized and implemented in this paper using multivariate differential equations. One or more paths of equilibria connects a city with perfectly elastic supply (where the equilibria are solutions of a polynomial system) to a city with a finite housing supply elasticity, and the end point of such path is an equilibrium of the general city. Such general city has social preferences over a geographic scope, pecuniary externalities through floor surface prices, heterogeneous amenities, and heterogeneous marginal costs.

\begin{figure}
    \caption{Homotopies from Simpler Cities to the City of Interest}
    \label{fig:enter-label}
    \begin{center}
    \begin{tikzcd}
       \mathcal{C}^\infty_h \ar[rd,"H_A"] \\
  & \mathcal{C}^\infty \ar[r,"H_\eta"] & \mathcal{C} \\
    \mathcal{P} \ar[ru,"H"]   &                                    
    \end{tikzcd}
    \end{center}
    \emph{$\mathcal{C}$ is the city of interest. It is a city with finite supply elasticity and an arbitrary vector of amenities.  $\mathcal{C}^\infty_h$ is a city with homogeneous amenities and perfectly elastic housing supply $\eta=\infty$. $\mathcal{C}^\infty$ is a city with an arbitrary vector of amenities and perfectly elastic housing supply. Each homotopy $H$ is a function of (a) the equilibrium vector and (b) a scaling parameter $t\in[0,1]$ that transforms the equilibrium system of equations smoothly from $t=0$ to $t=1$. $H_A$ is presented in Section~\ref{sec:first_homotopy}. $H_\eta$ is presented in Section~\ref{subsec:Homotopy}. $H_A$ does not experience bifurcations, however $H_\eta$ may experience bifurcations and a solution is presented in Section~\ref{sec:bifurcations}. Numerical exercises suggest that bifurcations may not be frequent. $\mathcal{P}$ is a polynomial system whose solutions are on the unit circle, as in \citeasnoun{kehoe1991computation}. This system is not typically the set of equations of a city, but the system of equations of a discrete choice model are homotopic to a polynomial. $H$ is the corresponding homotopy. $H$ is presented in Section~\ref{subsec:Perfectly-Elastic-Housing-Supply-Exact-Method}.} 
\end{figure}
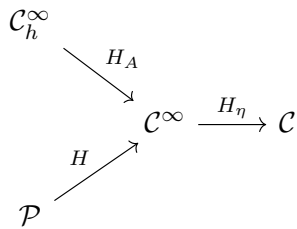

In the second homotopy from perfectly elastic supply to finite supply elasticity, a path of equilibria may bifurcate into two or more equilibria. At bifurcation points, a single path may branch out into multiple paths, potentially leading to multiple equilibria, and leading to multiple counterfactual cities for a single initial starting point of the path. We build upon an established literature in applied mathematics to follow these multiple paths and discover new equilibrium.  When a path, e.g. from a perfectly elastic city to a general city, splits as the elasticity of housing supply decreases, we can use solutions to a quadratic equation that is the second-order Taylor expansion of the homotopy to find the multiple paths leading to separate equilibria of the city~\cite{keller1977numerical}. The first-order condition gives us a subspace of dimension 1 or more, and we choose the solutions of the second-order Taylor expansion system that lie within such first-order condition.\footnote{A further option to deal with singularities  is to use arclength homotopy (also known as arclength continuation \cite{chan1982arc}), to follow turning points, bifurcation points, and cusps. Arclength continuation also has the benefit of overcoming the challenging of turning points and cusps by making the Jacobian invertible. Empirically however we do not see clear benefits of arclength homotopy in this urban economics setting.}

The classic approaches in applied mathematics lead to high order polynomials. In economics, we can exploit the specific structure of discrete choice models to reduce the dimensionality of the problem and to obtain shorter paths. We update and build on the method developed in \citeasnoun{kehoe1991computation} by replacing the initial polynomial by a polynomial with an exponentially smaller number of solutions. These solutions are  the equilibria of a city with homogeneous amenities. This exponentially speeds up calculation time. This paper uses homotopies, i.e. the transformation of equilibria from one model to another, which are the solutions of differential equations whose initial conditions are the equilibria of the first, simpler, city. When a path is regular, these paths yield an equilibrium of the city. A path is regular when the Jacobian of the differential equation is non-singular.

The approach is applied to thousands of numerical examples. The first set of examples performs the first homotopy. It enumerates the equilibria of cities with social interactions and perfectly elastic housing supply. Even with a small number of locations ($J=7$) the approach can yield 127 equilibria. We verify that the vectors of social demographics are indeed equilibria of the city. Intuitive properties emerge: a city with a small standard deviation of amenities has a larger number of equilibria, as geography does not ``pin down'' the precise equilibrium. The second set of examples starts with each of the 2,200+ cities of the first example to generate equilibria of the city with a finitely elastic housing supply. In this case, the response of floor surface prices is a source of strategic substitutability and tends to lower the number of equilibria. This second set of examples encompasses the workhorse models of spatial economics: its structure is similar to the quantitative spatial models used in the literature. 

The third set of examples pertains to the equilibria of the metropolitan area of Chicago. The paper puts forward a method that can deal with citywide homotopies with a large number of locations, here 353 neighborhoods, 77 communities, within 9 regions of the City. The approach is as follows. In the first step, the equilibria within communities and within regions are found focusing on social interactions only. In the second step, a citywide homotopy along the supply elasticity parameter is performed. Results generate counterfactual maps that are consistent with the estimated structural parameters. Parameters are estimated using a 1940-2010 panel of neighborhoods with consistent boundaries built using Census and Geolytics data.  These results suggest that there are a number of counterfactual Chicagos that are consistent with the estimated set of amenities.\footnote{We extend the Geolytics data set of \citeasnoun{card2008tipping} by adding three decades: 1940, 1950, 1960, and building Census tract relationship files.}

The paper also provides equilibria when unobservables are Fr\'echet-distributed, when unobservables are Gumbel-distributed as in \citeasnoun{mcfadden1977modelling}, in their static and dynamic versions.

The paper is implemented as follows. The code is written in Julia 1.9, and relies on the packages \texttt{ForwardDiff} for automatic differentiation \cite{revels2016forward}, \texttt{Differential} \texttt{Equations} for solving multivariate ordinary differential equations that give us the paths of equilibria, using predictor-corrector methods \cite{rackauckas2017differentialequations}, \texttt{HomotopyContinuation} when performing total degree homotopy for the solutions of arbitrary polynomials \cite{breiding2018homotopycontinuation}. These arbitrary polynomials occur in two places: when solving the first homotopy $H_A$, and at bifurcation points for the second-order Taylor expansion.

The paper is structured as follows. Section~\ref{sec:intuition} presents the intuition using a simple model. Section~\ref{sec:The-City} presents a general model of the city, where amenities are heterogeneous, multiple demographic groups choose locations, and location choices are driven both by social interactions and by price responses. Section~\ref{sec:frechet} presents the two-step homotopy method: first, finding all the equilibria for the city with perfectly elastic supply and heterogeneous amenities, whose equilibria are solutions of a polynomial system; second, transforming these equilibria into those of the city with elastic supply and heterogeneous amenities, by solving a differential equation with initial conditions equal to the first set of equilibria. Section~\ref{sec:examples} implements the method on more than a thousand Monte-Carlo examples, computing the set of equilibria for a range of cities with parameters in a multidimensional set. Section~\ref{sec:empirics} applies the method to  the City of Chicago, using parameters estimated from a panel data set with consistent tract boundaries.
Section~\ref{sec:discussion} discusses and extends the approach: first, to the case of a discrete choice model in the framework of \citeasnoun{mcfadden1977modelling} with the same ingredients as before; second, the section extends the approach to the dynamic case; in this case the equilibrium is the solution to a polynomial system in intertemporal welfare and social demographics. Third, Section~\ref{sec:bifurcations} presents the method used to deal with bifurcations.
  
\section{Intuition with A Simple Example}\label{sec:intuition}

Before considering the general case, we present the paper's main intuition
with a simple city with $J=3$ locations. We also use this example
to illustrate the approach that will be used for a large number of
locations $J$.

Each location is characterized by an amenity level $A_{1},A_{2},A_{3}$
and the price of floor surface in location $j$ is denoted $q_{j}$.
There is a mass $N$ of households choosing across the three locations,
$N^{g}$ college-graduates and $N-N^{g}$ non-college graduates. Denote
by $x_{j}=L_{j}^{g}/L^g$ the share of college educated population
at location $j$. Developers supply floor surface with elasticity
$\eta$, so that the supply of floor surface is $h_{j}=q_{j}^{\eta}$
in this simple example.

Households value consumption, floor surface, and the presence of college
graduates in nearby locations. In this simple example we assume that
$1,2,3$ are ordered on the line, and that households care about the
value of $\Psi_{j}$.
\begin{align*}
\Psi_{1} & =x_{1}+e^{-1}x_{2}+e^{-2}x_{3}\\
\Psi_{2} & =e^{-1}x_{1}+x_{2}+e^{-1}x_{3}\\
\Psi_{3} & =e^{-2}x_{1}+e^{-1}x_{2}+x_{3}
\end{align*}
As is typical in this literature, we assume as in \citeasnoun{benabou1996equity} and \citeasnoun{epple1999estimating} that college graduates have a greater preference for college graduates, $\gamma^{g}>\gamma$. This is the single-crossing condition. 

Each household maximizes its Cobb-Douglas utility with Fr\'echet
unobservables of dispersion~1. A college graduate chooses location
$j$ with probability:
\[
P_{j}^{g}=\frac{A_{1}q_{1}^{-\alpha}\Psi_{1}^{\gamma^{g}}}{A_{1}q_{1}^{-\alpha}\Psi_{1}^{\gamma^{g}}+A_{2}q_{2}^{-\alpha}\Psi_{2}^{\gamma^{g}}+A_{3}q_{3}^{-\alpha}\Psi_{3}^{\gamma^{g}}}
\]
College and non-college graduates only differ in their preferences
$\gamma^{g},\gamma$. The density of college graduates at location $j$
is then: $x_j = L_{j}^{g}/L^g= P_{j}^{g}$.

The price of floor surface $q_{j}$ clears each of the three markets
for floor surface:
\[
q_{j}=\left[N^{g}\frac{A_{1}\Psi_{1}^{\gamma^{g}}}{A_{1}q_{1}^{-\alpha}\Psi_{1}^{\gamma^{g}}+A_{2}q_{2}^{-\alpha}\Psi_{2}^{\gamma^{g}}+A_{3}q_{3}^{-\alpha}\Psi_{3}^{\gamma^{g}}}+(N-N^g)\frac{A_{1}\Psi_{1}^{\gamma}}{A_{1}q_{1}^{-\alpha}\Psi_{1}^{\gamma}+A_{2}q_{2}^{-\alpha}\Psi_{2}^{\gamma}+A_{3}q_{3}^{-\alpha}\Psi_{3}^{\gamma}}\right]^{1/(\alpha+\eta)}
\]
And the social equilibrium condition requires that the model-predicted
share of college graduates be consistent with the actual share.
\begin{align*}
x_{1} & = \frac{A_{1}q_{1}^{-\alpha}\Psi_{1}^{\gamma^{g}}}{A_{1}q_{1}^{-\alpha}\Psi_{1}^{\gamma^{g}}+A_{2}q_{2}^{-\alpha}\Psi_{2}^{\gamma^{g}}+A_{3}q_{3}^{-\alpha}\Psi_{3}^{\gamma^{g}}}\\
x_{2} & = \frac{A_{2}q_{2}^{-\alpha}\Psi_{1}^{\gamma^{g}}}{A_{1}q_{1}^{-\alpha}\Psi_{1}^{\gamma^{g}}+A_{2}q_{2}^{-\alpha}\Psi_{2}^{\gamma^{g}}+A_{3}q_{3}^{-\alpha}\Psi_{3}^{\gamma^{g}}}\\
x_{3} & = \frac{A_{3}q_{3}^{-\alpha}\Psi_{3}^{\gamma^{g}}}{A_{1}q_{1}^{-\alpha}\Psi_{1}^{\gamma^{g}}+A_{2}q_{2}^{-\alpha}\Psi_{2}^{\gamma^{g}}+A_{3}q_{3}^{-\alpha}\Psi_{3}^{\gamma^{g}}}
\end{align*}
At this stage we are ready to find the equilibria of this city. The
first step is start by finding the equilibria of the city with perfectly
elastic housing supply $\eta\rightarrow\infty$. In this case, $q_{1}=q_{2}=q_{3}=1$.
The city's equilibrium conditions can be written as:
\begin{align*}
\Psi_{1} & = \frac{A_{1}\Psi_{1}^{\gamma^{g}}}{A_{1}\Psi_{1}^{\gamma^{g}}+A_{2}\Psi_{2}^{\gamma^{g}}+A_{3}\Psi_{3}^{\gamma^{g}}}\\
 & +e^{-1} \frac{A_{2}\Psi_{2}^{\gamma^{g}}}{A_{1}\Psi_{1}^{\gamma^{g}}+A_{2}\Psi_{2}^{\gamma^{g}}+A_{3}\Psi_{3}^{\gamma^{g}}}+e^{-2} \frac{A_{3}\Psi_{3}^{\gamma^{g}}}{A_{1}\Psi_{1}^{\gamma^{g}}+A_{2}\Psi_{2}^{\gamma^{g}}+A_{3}\Psi_{3}^{\gamma^{g}}}
\end{align*}
for $j=1$ and similarly for $j=2$ and $j=3$. The key insight is
that this can be expressed as a polynomial in integer powers. First approximate the social preference parameter $\gamma^g$ by a rational fraction:
\begin{equation}
    \gamma^g \simeq \frac{p}{q}
\end{equation}
Then perform the change of variables $z_{j}=\Psi_{j}^{1/q}$ for each location $j$.
Then the city's equilibrium is a system of polynomials in $(z_{1},z_{2},z_{3})$
of order $p+q$. The polynomial system is:
\begin{equation}
\left\{
\begin{array}{lcl}
A_{1}z_{1}^{p}z_{1}^{q}+A_{2}z_{2}^{p}z_{1}^{q}+A_{3}z_{3}^{p}z_{1}^{q}- A_{1}z_{1}^{p}- e^{-1}A_{2}z_{2}^{p}- e^{-2}A_{3}z_{3}^{p} & = & 0 \\ 
A_{1}z_{1}^{p}z_{2}^{q}+A_{2}z_{2}^{p}z_{2}^{q}+A_{3}z_{3}^{p}z_{2}^{q}- e^{-1}A_{1}z_{1}^{p}- A_{2}z_{2}^{p}- e^{-1}A_{3}z_{3}^{p} & = & 0 \\
A_{1}z_{1}^{p}z_{3}^{q}+A_{2}z_{2}^{p}z_{3}^{q}+A_{3}z_{3}^{p}z_{3}^{q}- e^{-2}A_{1}z_{1}^{p}- e^{-1}A_{2}z_{2}^{p}- A_{3}z_{3}^{p} & = & 0 
\end{array}
\right.
\label{eq:example_poly_system_equation_1}
\end{equation}
The equilibria of this city lie at the intersection of these three polynomia in $z_1,z_2,z_3$. These intersection can be found by `smoothly' transforming the amenity coefficient $A_j$ and finding a path of $(z_1,z_2,z_3)\in \mathbb{R}^3$ from a simpler system of $A_j$s to a system with the city's actual $A_j$s.

\subsubsection*{First Approach: From the Unit Circle to the City's Equilibrium Equations\footnote{While it may seem that such system (\ref{eq:example_poly_system_equation_1}) should have closed form solutions, approaches for solving systems of polynomials suggest that this is in general a difficult problem, even with a small number of locations: for instance the Groebner base decomposition of the polynomial with 3 locations yields a polynomial in $z_1$ of degree higher than 70. This paper provides a more parsimonious approach.}}

To solve this, we use a simpler polynomial for which we know all the solutions exactly. This is a common approach in algebraic geometry, called total degree homotopy. Write the system as a linear combination in
$t\in[0,1]$ of a simpler system:
\begin{equation}
\left\{
\begin{array}{lcl}
    z_{1}^{p+q}&=&0 \\
    z_{2}^{p+q}&=&0 \\
    z_{3}^{p+q}&=&0
\end{array} \right.  \label{eq:unit_circle_system}
\end{equation}
and the city's actual polynomial system. There solution $(z_1,z_2,z_3)$ for each convex combination of this polynomial system (\ref{eq:unit_circle_system}) and the polynomial system of interest (\ref{eq:example_poly_system_equation_1}).
\begin{equation}
    H(z_1,z_2,z_3,t) = 0 
\end{equation}
This is a system of
three equations, the first one being:
\[
tz_{1}^{p+q}+(1-t)\left[A_{1}z_{1}^{p}z_{1}^{q}+A_{2}z_{2}^{p}z_{1}^{q}+A_{3}z_{3}^{p}z_{1}^{q}-N^{g}A_{1}z_{1}^{p}-N^{g}e^{-1}A_{2}z_{2}^{p}-N^{g}e^{-1}A_{3}z_{3}^{p}\right]=0
\]
Now the simpler system ($t=1$) has as roots the points on the complex unit
circle $e^{2i\pi n/(p+q)}$, which is $(p+q)^{3}$ potential solutions. From one of these starting points, a solution $(z_1,z_2,z_3)$ for $t=1$, a solution for $t<1$ is found by integrating $d\mathbf{z}/dt$ along $\tau \in [t,1]$. 

This approach guarantees that we enumerate all equilibria by
a key theorem of algebraic geometry \cite{sommese2005numerical}: each path from a root on the unit circle either diverges or converges to a solution of the polynomial system. 

When $t=1$, these are the roots of the system (\ref{eq:unit_circle_system}). When $t=0$, the roots
are potential equilibria of the city, when they are real and $\Psi_{j}\in[0,1]$.
The key idea of total degree homotopy continuation is to keep track
of the roots as $t\rightarrow t+dt$, calculating $dz_{j}/dt$ and
solving the differential equation by a solver such as Runge Kutta,
Adams-Bashforth and Adams-Moulton, or other differential equation
solving techniques. The solutions found for $t=0$ are depicted on Figure~\ref{fig:A-Visual-Example} as the blue points.
We check that the solutions are precise to $10^{-16}$. We also check
that these solutions satisfy the population condition.

\begin{figure}

\caption{A Visual Example of Equilibrium Multiplicity with $J=3$ Locations\label{fig:A-Visual-Example}}

\bigskip

\emph{This figure with $J=3$ locations illustrates the paper's baseline
approaches to equilibrium multiplicity. We obtain 5 equilibria in
both the city with perfectly elastic housing supply (blue points)
and in the city with finite elasticity of housing supply (red points). The blue
points are obtained by finding the roots of the polynomial system
(\ref{eq:example_poly_system_equation_1}) in $(z_{1},z_{2},z_{3})$
where $\mathbf{z}=\boldsymbol{\Psi}^{1/q}$ and $\boldsymbol{\Psi}=\Delta\mathbf{x}$.
The variables of interest are the densities of college graduates $(x_{1},x_{2},x_{3})$
in each location, $\mathbf{x}=\Delta^{-1}\mathbf{z}^{q}$, which is
what is depicted here. The red points are obtained by smoothly changing
the equilibria (blue points) from an infinitely elastic supply ($\eta=\infty$)
to an elastic supply (here $\eta=0.1$). These are the paths, solutions
of a differential equation where the initial conditions are the blue
points.}

\begin{center}

\includegraphics[scale=0.6]{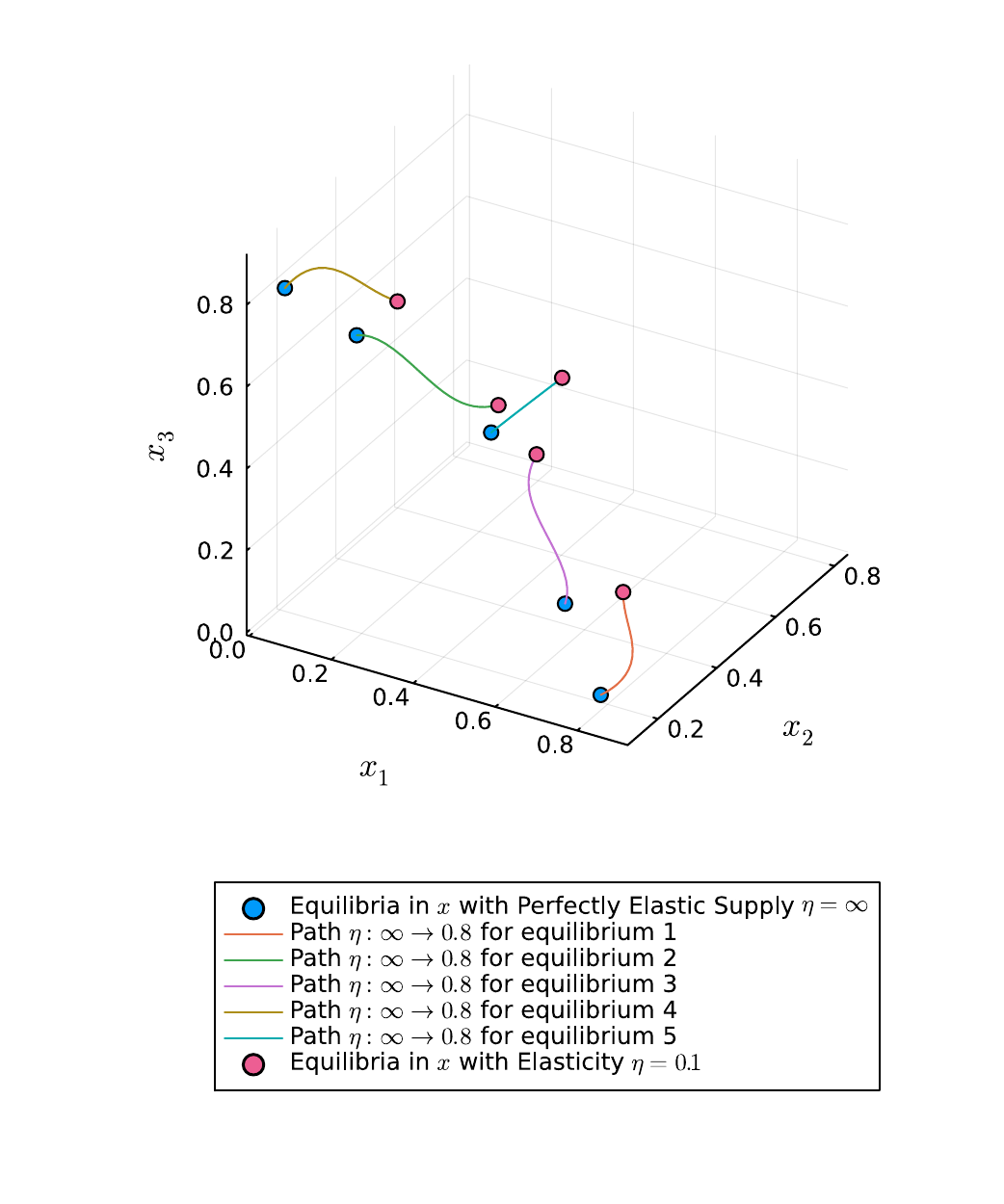}

\end{center}

\emph{Parameters: $\alpha=0.3$, $\gamma^{g}=2.5$ (hence $p=5$,
$q=2$), $\gamma=0$, $A_{1}=A_{2}=A_{3}=1$, $N^{g}=0.8*3$, $N=3$,
$\delta_{ij}=\textrm{exp}(-|i-j|)$, $\eta=0.1$ or $\eta=\infty$.}
\end{figure}

\subsubsection*{Second Approach: From the City with Identical Locations to the Heterogeneous City}

This first approach can be infeasible when the number of locations is substantial. We develop a second more parsimonious approach that
reduces the dimensionality of the problem.  This approach has lower dimensionality $q^3$.

We build a homotopy from the city with homogeneous
amenities and weights to the city with heterogeneous amenities $A_{1},A_{2},A_{3}$,
and heterogeneous social interaction weights $\Delta_{ij}$. The system:
\begin{equation}
z_{1}^{p}z_{1}^{q}+z_{2}^{p}z_{1}^{q}+z_{3}^{p}z_{1}^{q}-e^{-1}z_{1}^{p}-e^{-1}z_{2}^{p}-e^{-1}z_{3}^{p}=0,\label{eq:example_poly_system_equation_1-1}
\end{equation}
simplifies to $z_{1}^{q} = 3 e^{-1}$ and has a known set of real and complex roots. Set:
\[
A_{1}(t)=(1-t)+tA_{1},\,\,A_{2}(t)=(1-t)+tA_{2},\,\,A_{3}(t)=(1-t)+tA_{3}
\]
for $t$ from $0$ (homogeneous amenities) to $t=1$ (heterogeneous
amenities). And perform the same transformation for $\Delta_{ij}$
and the marginal cost $mc$.

\subsubsection*{Allowing for a Price Response: Towards the City with Elastic Supply}

The next step is to allow for an elasticity of housing supply set
to $h_{j}=q_{j}^{\eta}$ in this simple example. We started with $\eta=+\infty$.
$[0,\infty)$ is not a compact set, hence a change of variable will
help here. We define $\zeta=1/\eta$, which allows following the solution
on $\zeta\in[0,1/\eta]$, and the paper shows that the differentials
do not diverge at $\zeta=0$ and therefore can be extended by continuity at $\zeta=0$.
This time we follow both the evolution of demographics $d\log\mathbf{z}/d\eta$
and the evolution of prices $d\log\mathbf{q}/d\eta$ along the path, and integrate them. This is
depicted on Figure~\ref{fig:A-Visual-Example}.

\subsubsection*{Existence of Paths}

The existence of a path can be jeopardized by the singularity of the Jacobian of the homotopy continuation. 
In intuitive terms, there may be multiple solutions  $(dz_1,dz_2,dz_3)$ compatible with a small change $dt$ at a singular point.  Intuitively, this occurs when the polynomial equations share the same tangent space. This can be visualized in the case of 2 locations on Figure~\ref{fig:two-location-case-geometry}(c). Figure~(b) shows the existence of 3 equilibria in this city with 2 locations and homogeneous amenities $A_1=A_2$. As we make amenities more heterogeneous, i.e. we increase the amenity of location 2 relative to the amenity of location 1, Figure~(c) shows the presence of a singular point, where the tangent to the first polynomial equation is equal to the tangent of the second polynomial equation. For higher values of $A_2$, the city has a \emph{single} equilibrium, where residents choose location 2. For lower values of $A_2$, the city has three equilibria. For the value of $A_2$ used on Figure~(c), the two polynomial curves are tangent, and equilibrium 2 has multiplicity 2. Figure~\ref{fig:three-location-case-geometry} suggests that this intuition extends to the case of 3 locations.

We use a systematic approach. The singular point can be found by finding the $t\in[0,1]$ at which the determinant or the smallest eigenvalue of the Jacobian cancel. At this singular point, the bifurcations are computed as follows \cite{keller1977numerical}. First we find the kernel of the Jacobian matrix. This gives a subspace of dimension $\geq 1$. We then write the second order Taylor expansion at this singular point.  This is a polynomial system of order 2, for which we can apply a homotopy to find the finite number of potential bifurcations at the singular point that lie within the kernel. 

\section{The City\label{sec:The-City}}

We consider a city with $J$ locations characterized by amenities
$A_{j}$. There is a population $L_{1}$ of type 1 households and
$L_{2}$ of type 2 households, and social preferences $\gamma^{1}$
and $\gamma^{2}$ w.r.t to the demographic composition in the neighborhood
and in neighboring areas. The price of housing is $q_{jt}$, and the
elasticity of demand w.r.t. housing is $\alpha$. The indirect utility
of location $j$ is:
\begin{equation}
V_{j}^{g}=A_{j}q_{j}^{-\alpha}\Psi_{j}^{\gamma^{g}}\varepsilon_{ij}\label{eq:utility_function}
\end{equation}
where $\varepsilon_{ij}$ is Frechet-distributed. The dispersion parameter
of set to $\theta=1$ to lighten the notations without loss of generality.
The term in $\Psi$ measures the scope of social interactions as a
weighted average of the population density $x_{j}$ of type 1 in neighborhoods
at distance $d_{jk}$ of neighborhood $j$:
\begin{equation}
\Psi_{j}=\sum_{k=1}^{J}e^{-\xi d_{jk}}x_{k}\label{eq:definition_of_psi}
\end{equation}
where $d_{jj} \equiv 0$. The scalar $\xi$ measures the scope of social interactions \cite{redding2017quantitative}: when $\xi$ is large, the household cares only about the demographics of the current location. Define $\Delta$ the matrix $\left[e^{-\xi d_{jk}}\right]_{j,k=1,2,\dots,J}\in(0,\infty)^{J^{2}}$
. Then the vector $\mathbf{\Psi}=\Delta\mathbf{x}$. We assume that
$e^{-\xi d_{jj}}>\sum_{k\neq j}e^{-\xi d_{jk}}$ for all $j,k=1,2,\dots,J$
which ensures that $\Delta$ is invertible as a diagonally dominant
matrix. 

The probability of choosing location $j$ is thus:
\begin{equation}
P_{j}^{g=1}=\frac{A_{j}q_{j}^{-\alpha}\Psi_{j}^{\gamma^{1}}}{\sum_{k}A_{k}q_{k}^{-\alpha}\Psi_{k}^{\gamma^{1}}}\label{eq:probability_of_choosing_location_j}
\end{equation}
and thus the demographic composition of location $j$ is:
\begin{equation}
x_{j}=\frac{L^1_j}{L_{1}}\frac{A_{j}q_{j}^{-\alpha}\Psi_{j}^{\gamma^{1}}}{\sum_{k}A_{k}q_{k}^{-\alpha}\Psi_{k}^{\gamma^{1}}}\label{eq:social_equilibrium}
\end{equation}
Developers supply housing competitively with elasticity $\eta$.
\[
h_{j}=c_{j}\left(\frac{q_{j}}{\text{mc}_{j}}\right)^{\eta}
\]
so that $q_j = \text{mc}_j$ when $\eta \rightarrow \infty$. 

At equilibrium in each floor surface market:
\begin{equation}
\frac{1}{\text{mc}_{j}^{\eta}}c_{j}q_{j}^{\alpha+\eta}= L_1 \frac{A_{j}\Psi_{j}^{\gamma^{1}}}{\sum_{k}A_{k}q_{k}^{-\alpha}\Psi_{k}^{\gamma^{1}}}+ L_2 \frac{A_{j}\Psi_{j}^{\gamma^{2}}}{\sum_{k}A_{k}q_{k}^{-\alpha}\Psi_{k}^{\gamma^{2}}}\label{eq:market_clearing_conditions}
\end{equation}
An equilibrium is a pair of social demographics and floor surface
prices such that the market in each $j$ clears, and the social composition
predicted by the model is equal to that perceived by households.

The equilibrium can be expressed in terms of $\Psi_{j}$ and $q_{j}$
only instead of $x_{j}$ and $q_{j}$. Equation \ref{eq:social_equilibrium}
can be rewritten as:
\begin{equation}
\Psi_{j}=\sum_{k=1}^{J}e^{-\xi d_{jk}}x_{k}=\sum_{k=1}^{J}e^{-\xi d_{jk}} \frac{A_{k}q_{k}^{-\alpha}\Psi_{k}^{\gamma^{1}}}{\sum_{l}A_{l}q_{l}^{-\alpha}\Psi_{l}^{\gamma^{1}}}\label{eq:social_equilibrium_conditions}
\end{equation}
Together, the social equilibrium conditions (\ref{eq:social_equilibrium_conditions}) and the market-clearing conditions (\ref{eq:market_clearing_conditions}) define an equilibrium vector.

\begin{defn}
\textbf{(City Equilibrium)\label{def:(City-Equilibrium)-An}}\emph{
A }\textbf{city}\emph{ $\mathcal{C}=(\mathbf{L},\mathbf{A},\Delta,\mathbf{c},mc,\eta,\boldsymbol{\gamma},\alpha)$
is an aggregate population $\mathbf{L}=(L_{g})_{g=1,2}$ for each demographic
group, a vector of amenities $\mathbf{A}\in\mathbb{R}_{+*}^{J}$ for
each location, a distance matrix $\Delta\in\mathbb{R}^{J^{2}}$, a
vector of housing supply constants $\mathbf{c}\in\mathbb{R}_{+*}^{J}$,
a vector of marginal costs $mc\in\mathbb{R}_{+*}^{J}$, an elasticity
of housing supply $\eta$, a scalar $(\gamma^{g})_{g=1,2}\in\mathbb{R}_{+*}^{2}$
of preferences for each group $g=1,2$, an elasticity of demand $\alpha\in\mathbb{R}_{+*}$
w.r.t. the price. For each city $\mathcal{C}$, denote by $\mathcal{C}^{\infty}=(\mathbf{L},\mathbf{A},\Delta,\mathbf{c},mc,\infty,\boldsymbol{\gamma},\alpha)$
the city for which $\eta=\infty$ and all other parameters are kept
constant. Denote by $\mathcal{C}_{h}=(\mathbf{1},\mathbf{1}_{J},e\mathbf{1}_{J\times J},\mathbf{1}_{J},\mathbf{1}_{J},\eta,\boldsymbol{\gamma},\alpha)$
the city with unit population, homogeneous amenities $\mathbf{A}=1_{J}$
and a distance matrix $\Delta=e1_{J\times J}$ of ones times the Euler
constant, homogeneous supply constants and marginal costs, and all
other parameters are kept constant.}

\emph{A }\textbf{proper}\emph{ }\textbf{equilibrium }\emph{$(\mathbf{x}^{*},\mathbf{q}^{*})\in[0,1]^{J}\times\mathbb{R}_{+*}^{J}$}\textbf{
of the city $\mathcal{C}$}\emph{ is a vector $\mathbf{x}^{*}\in[0,1]^{J}$
of demographic compositions and a vector $\mathbf{q}^{*}\in(0,\infty]^{J}$
of prices that satisfy the $J$ market clearing conditions~(\ref{eq:market_clearing_conditions})
and the $J$ social equilibrium conditions~(\ref{eq:social_equilibrium}).}

\emph{Equivalently, a }\textbf{proper equilibrium}\emph{ is a pair
$(\boldsymbol{\Psi}^{*},\mathbf{q}^{*})\in[0,\infty)^{J}\times(0,\infty]^{J}$
that satisfy the $J$ market clearing conditions~(\ref{eq:market_clearing_conditions}),
the $J$ social equilibrium conditions (\ref{eq:social_equilibrium_conditions}), stacked in a single function of dimension $2J$:
\begin{equation}
E(\Psi^{*},\mathbf{q}^{*})=0_{2J}\label{eq:equilibrium_conditions}
\end{equation}
and for which $\mathbf{x}=\Delta^{-1}\boldsymbol{\Psi}\in[0,1]^{J}$.
$E$ is a continuous and infinitely differentiable function for strictly positive values of $(\Psi^{*},\mathbf{q}^{*})$.
An }\textbf{improper equilibrium}\textbf{\emph{ }}\emph{is a pair
$(\boldsymbol{\Psi}^{*},\mathbf{q}^{*})\in\mathbb{R}^{J}\times(0,\infty]^{J}$
such that $(\boldsymbol{\Psi}^{*},\mathbf{q}^{*})$ satisfies the
equilibrium relationships but for which $\mathbf{x}^{*}=\Delta^{-1}\mathbf{\Psi}^{*}\notin[0,1]^{J}$.
The set $G(\mathcal{C})\subset\mathbb{R}^{J}\times(0,\infty]^{J}$
of }\textbf{generalized equilibria}\textbf{\emph{ }}\emph{is the union
of the set of proper and improper equilibria of the city.}
\end{defn}

\section{Equilibria of the City using Homotopies}\label{sec:frechet}

We make three remarks. 

First, it is possible to guarantee the enumeration of all equilibria
of the city $\mathcal{C}^{\infty}$ with heterogeneous amenities and
perfectly elastic housing supply. The equilibrium equations of the
economy $E$ can be expressed as the solutions of a polynomial system
with integer powers. Standard results in the study of polynomial systems
guarantee that all solutions of such polynomial systems can be enumerated.
This is because such a polynomial system is homotopic to a system
of polynomials whose solutions are on the complex unit circle. We
provide a quantitative application of this method on a large set of
possible cities.

Yet, this approach can be computationally intensive and infeasible
for larger cities. This leads to our second approach: the city $\mathcal{C}^{\infty}$
is homotopic to the corresponding city $\mathcal{C}_{h}^{\infty}$
with homogeneous amenities. In such a city all real solutions of the
equilibrium system are known. 

The third and final remark is that the equilibria of the economy with
a finite elasticity of housing supply $\eta^{*}<\infty$ can be found
by smoothly transforming the equilibria of an economy with perfectly
elastic housing supply, $\eta=\infty$, along the path $\eta\in(\infty;\eta^{*}]$.
The model behaves continuously as $\eta\rightarrow\infty$, which
enables a change of variable $\zeta=1/\eta\in[0,\frac{1}{\eta^{*}}]$. 

This section is structured as follows. Section~\ref{subsec:Homotopic-Cities}
introduces the definition of a homotopy and the approach using differential
equations. Section~\ref{subsec:Perfectly-Elastic-Housing-Supply-Exact-Method}
derives the equilibria of the city $\mathcal{C}^{\infty}$ with perfectly
elastic supply, $\eta=\infty$. Section~\ref{sec:Finite-Housing-Supply-Elasticity}
then smoothly transforms the equilibria $(\boldsymbol{\Psi}^{*}(\eta=\infty),\mathbf{q}^{*}(\eta=\infty))$
of city $\mathcal{C}^{\infty}$ to the equilibria of the city $\mathcal{C}$
$(\boldsymbol{\Psi}^{*}(\eta^{*}),\mathbf{q}^{*}(\eta^{*}))$ with
a finite elasticity of housing supply. This concludes the approach. 

\subsection{Homotopic Cities\label{subsec:Homotopic-Cities}}
\begin{defn}
\textbf{(Homotopic systems of equations, Homotopic cities)}\emph{
Two continuous functions $E:\mathbb{R}^{J}\times\mathbb{R}^{J}\rightarrow\mathbb{R}$
and $E':\mathbb{R}^{J}\times\mathbb{R}^{J}\rightarrow\mathbb{R}$
are homotopic if there exists a family of continuous transformation
indexed by $t\in[0,1]$:
\begin{align}
H: & \mathbb{R}^{J}\times\mathbb{R}^{J}\times[0,1]\rightarrow\mathbb{R}\nonumber \\
 & (\boldsymbol{\Psi},\mathbf{q},t)\mapsto H(\boldsymbol{\Psi},\mathbf{q},t)
\end{align}
such that $H(\Psi^{*},\mathbf{q}^{*},0)=E(\Psi^{*},\mathbf{q}^{*})$
and $H(\Psi^{*},\mathbf{q}^{*},1)=E'(\Psi^{*},\mathbf{q}^{*})$ for
all $(\boldsymbol{\Psi}^{*},\mathbf{q}^{*})\in\mathbb{R}^{J}\times\mathbb{R}^{J}$
. A differentiable homotopy occurs when $H$ is differentiable with
respect to $(\boldsymbol{\Psi},\mathbf{q},t)$.}

\emph{By extension, two cities $\mathcal{C}$ and $\mathcal{C}'$
are homotopic if their systems of equilibrium equations $E$ and $E'$
are homotopic. }
\end{defn}

Homotopies are particularly useful if one can exactly calculate the equilibria of the initial ($t=0$) city with pen and paper, and then use a path from each equilibrium to potentially find an equilibrium of the city with heterogeneous amenities, marginal costs, supply elasticity, distance matrix. The properties of homotopies can be challenging to study in the non-linear general case, but we will show in the next Section \ref{subsec:Perfectly-Elastic-Housing-Supply-Exact-Method} that the system of equilibrium equations can be expressed as a polynomial system. This allows us to study the following differential equation along the path $t\in[0,1]$:

\begin{prop}
\textbf{\emph{(Equilibrium Paths from $\mathcal{C}$ to $\mathcal{C'}$)}} Consider two homotopic cities
$\mathcal{C}$ and $\mathcal{C}'$ with differentiable homotopy $H$
and a generalized equilibrium $(\boldsymbol{\Psi}^{*},\mathbf{q}^{*})\in\mathbb{R}^{J}\times(0,\infty]^{J}$
of city $\mathcal{C}$. Define the differential equation:
\begin{equation}
\frac{\partial H}{\partial(\Psi,\mathbf{q})}\frac{d(\Psi,\mathbf{q})}{dt}=-\frac{\partial H}{\partial t}, \label{eq:ode}
\end{equation}
with initial condition $\boldsymbol{\Psi}(t=0)=\boldsymbol{\Psi^{*}}$
and $\mathbf{q}(t=0)=\mathbf{q}^{*}$. We assume that the Jacobian $\frac{\partial H}{\partial(\Psi,\mathbf{q})}$
of $H$ with respect to demographic composition and price is invertible and continuous on a compact set including the path $(\boldsymbol{\Psi}(t),\mathbf{q}(t))$. 
The solution to this differential equation at $t=1$ is an equilibrium
of the city $\mathcal{C}'$.
\end{prop}

Each row of $-\left( \frac{\partial H}{\partial(\Psi,\mathbf{q})} \right)^{-1} \frac{\partial H}{\partial t}$ is a rational fraction as a fraction of two polynomials. The zeros of the denominator are the singular points of the Jacobian $\frac{\partial H}{\partial(\Psi,\mathbf{q})}$. On a bounded compact set where the Jacobian is invertible, the rational fraction is Lipschitz continuous. The Lipschitz continuity of the right-hand side implies the existence of a solution by Picard's theorem (\possessivecite{kolmogorov1975introductory} Theorem 2). 

In all examples that follow, $E$, $E'$, $H$ are polynomial systems with $J$ unknowns, and we test for singularities, which implies Lipschitz continuity of the right-hand side of (\ref{eq:ode}) outside of open sets around singularities. 

Characterizing the critical points of the Jacobian of an arbitrary polynomial system is an open research question; yet for each example it is empirically possible to characterize the set of critical points, which typically has low dimensionality as the zeros of the determinant. The equation $\textrm{det}(\frac{\partial H}{\partial(\Psi,\mathbf{q})})=0$ defines a $J-1$ curve in the $J$ dimensional space; for 3 locations, this is a surface.\footnote{In the important case of a general equilibrium Arrow-Debreu system, \citeasnoun{balasko2009equilibrium} provides results on the dimension and measure of this set.} Section~\ref{sec:singularity_jacobian} estimates the set of critical points of the Jacobian $\frac{\partial H}{\partial(\Psi,\mathbf{q})}$, its dimension and its impact on homotopy paths. 

\subsection{Enumerating the Equilibria of the City with Perfectly Elastic
Supply:\protect \\
Bertini's Approach, Homotopy $H$ \label{subsec:Perfectly-Elastic-Housing-Supply-Exact-Method}}

We consider a city $\mathcal{C}^{\infty}$ with perfectly elastic
housing supply. The problem of finding social equilibria $\Psi\in\mathbb{R}^{J}\backslash\{0,0,\dots,0\}$
is the problem of finding solutions to the system of $J$ equations:
\begin{equation}
\Psi_{j}=\sum_{k=1}^{J}e^{-\xi d_{jk}}\frac{A_{k}\text{mc}_{k}^{-\alpha}\Psi_{k}^{\gamma^{1}}}{\sum_{l}A_{l}\text{mc}_{l}^{-\alpha}\Psi_{l}^{\gamma^{1}}}\label{eq:equilibrium_psi}
\end{equation}
where $\textrm{mc}_j$ is the marginal cost of floor surface and $e^{-\xi d_{jk}}$ is the $j,k$ element of the distance matrix $\Delta$.
Here we exclude the trivial solution $\boldsymbol{\Psi}=\{0,0,\dots,0\}$.
The social demographic vector $\mathbf{x}$ is $\mathbf{x}=\Delta^{-1}\boldsymbol{\Psi}$. 

The set of rational numbers $\mathbb{Q}$ is dense in the set of real
numbers $\mathbb{R}$.\footnote{See for instance \citeasnoun{kolmogorov1975introductory}.} Hence, given an arbitrary $\varepsilon>0$,
there always exists a fraction $p/q$, with $p\in\mathbb{Z}$ and
$q\in\mathbb{N}$ at distance at most $\varepsilon$ of $\gamma^{1}\in\mathbb{R}$. Consider $p$ and $q$ with no common divisor.

\begin{equation}
\forall\varepsilon>0,\,\exists(p,q)\in\mathbb{Z}\times\mathbb{N}\qquad\left|\gamma^{1}-\frac{p}{q}\right|<\varepsilon\label{eq:rational_dense_in_R}
\end{equation}
Perform the change of variables $z_{j}\equiv\Psi_{j}^{1/q}$. The
equilibrium equation (\ref{eq:equilibrium_psi}) becomes:
\begin{equation}
z_{j}^{q}=\sum_{k=1}^{J}\Delta_{jk} \frac{A_{k}\text{mc}_{k}^{-\alpha}z_{k}^{p}}{\sum_{l}A_{l}\text{mc}_{l}^{-\alpha}z_{l}^{p}}\label{eq:zjq}
\end{equation}
which can be written as a polynomial system of $J$ polynomials of
order $p+q$:

\begin{equation}
\sum_{l=1}^{J}A_{l}\text{mc}_{l}^{-\alpha}z_{l}^{p}z_{j}^{q}-\sum_{k=1}^{J}\Delta_{jk} A_{k}\text{mc}_{k}^{-\alpha}z_{k}^{p}=0\label{eq:polynomial_system}
\end{equation}
while focusing on the solutions for which at least one $z_{j}\neq0$.
This condition ensures that equation (\ref{eq:equilibrium_psi}) is
well-posed and that the population condition is satisfied. This polynomial
system implies:
\begin{prop}
\emph{(}\textbf{\emph{Equilibria and
Solutions to the Polynomial System}}\emph{)} In the city $\mathcal{C}^{\infty}$
with perfectly elastic housing supply, every real and strictly positive solution
$\mathbf{z}^{*}\in\mathbb{R}_{+*}^{J}$ to the system of $J$ polynomial
equations of degree $p+q$, equations~(\ref{eq:polynomial_system})
can be mapped to a unique equilibrium $(\boldsymbol{\Psi}^{*},\mathbf{q}^{*})\in G(\mathcal{C}^{\infty})$
by setting $\boldsymbol{\Psi}^{*}=(\mathbf{z}^{*})^{q}$ and $\mathbf{q}^{*}=\mathbf{mc}$.
Conversely, every equilibrium $(\boldsymbol{\Psi}^{*},\mathbf{q}^{*})\in G(\mathcal{C}^{\infty})$
corresponds to a real and positive solution to the set of polynomial
equations of degree $p+q$. This set of $J$ polynomial equations (\ref{eq:polynomial_system}) is denoted as $P(\mathbf{z})=0$. 
\end{prop}

This polynomial system can be solved using homotopy continuation,
by using the following theorem from \citeasnoun{sommese2005numerical}, which is
an application of results by Bertini and B\'ezout. Such theorem guarantees
that we can obtain all solutions of the system almost surely by starting
with roots of a polynomial system whose roots are on the complex unit
circle. This theorem considers the case of a simpler polynomial of
order $p+q$ for which the solutions are known as points on the complex
unit circle. We then transform this polynomial system smoothly into
the system~\ref{eq:polynomial_system}. 

\begin{thm}
\label{thm:theorem_total_degree_homotopy}\textbf{\emph{(Equilibria
by Total Degree Homotopy)}}\emph{ }The polynomial system $P_{C^{\infty}}$
of equilibrium equations (\ref{eq:polynomial_system}) of the city
$\mathcal{C}^{\infty}$ with perfectly elastic housing supply is homotopic
to the polynomial system:
\begin{equation}
z_{j}^{d_{j}}-1=0,\qquad j=1,2,\dots,J\label{eq:simple_polynomial_system_total_degree_homotopy}
\end{equation}
where $d_{j}$ is the degree of equation $j$. Its solutions are the
complex roots $(z_{j})_{j=1,2,\dots,J}=(e^{2i\pi n_{j}/(p+q)})_{j}$
on the unit circle. $(n_{j})$ is a vector of integers $n_{j}=0,1,2,\dots,d_{j}$.
Denote by $G(\mathbf{z})=0$ this system. The $p+q$ solution paths
of the homotopy:
\begin{equation}
H(\mathbf{z},t)=\gamma tG(\mathbf{z})+(1-t)P(\mathbf{z}),\qquad t\in[0,1]\label{eq:homotopy}
\end{equation}
starting at the solutions of $G(\mathbf{z})=0$ are nonsingular for
$t\in(0,1]$ and their endpoints as $t\rightarrow0$ include \emph{all
of the nonsingular solutions} of $P(\mathbf{z})$ for almost all $\gamma\in\mathbb{C}$. 
\end{thm}

The proof of the second part of this theorem is in \citeasnoun{sommese2005numerical}.

Once a solution $z^{*}\in\mathbb{R}^{J}$ is found, the equilibrium
social demographics $x_{j}$ in each location are obtained using the
matrix relationship:
\begin{equation}
\mathbf{x}^{*}=\Delta^{-1}\boldsymbol{\Psi}^{*}=\Delta^{-1}(\mathbf{z}^{*q})\label{eq:transformation_psi_into_x}
\end{equation}
with $\Delta=\left[e^{-\xi d_{jk}}\right]_{j,k=1,\dots,J}$.

Theorem \ref{thm:theorem_total_degree_homotopy} guarantees that we
will find all equilibria of the economy by starting from the set of
complex roots of the unit circle. 

\subsection{Enumerating the Equilibria of the City with Perfectly Elastic
Supply:\protect \\
A Homotopy $H_A$ from a City with Homogeneous Amenities}\label{sec:first_homotopy}

This method can be computionally infeasible and we develop a heuristic that empirically yields a large number of equilibria. 
Such approach is closer to the final system
of equations and thus generates shorter paths, a significant numerical
benefit for homotopies \cite{morgan1987homotopy}. The following
proposition is central to this paper's approach.
\begin{prop}
\label{prop:The-generalized-equilibria}\textbf{\emph{(City with homogeneous
locations) }}The generalized equilibria of the city with homogeneous
amenities $\mathcal{C}_{h}^{\infty}=(\mathbf{L},\mathbf{1}_{J},e\mathbf{1}_{J\times J},mc,\infty,\boldsymbol{\gamma},\alpha)$
are the $2^{J}$ real vectors:
\begin{equation}
\mathbf{z}^{*}=\left((-1)^{d_{1}^{k}}e^{-2}\,(-1)^{d_{2}^{k}}e^{-2}\,\dots\,(-1)^{d_{J}^{k}}e^{-2}\right)\label{eq:generalized_equilibria_polynomial}
\end{equation}
where $d_{j}^{k}\in\{0,1\}$ is the binary representation of the integer
$k$ between $0$ and $2^{J}-1$.
\end{prop}

Hence if we can transform by homotopy these equilibria into equilibria
of the city $\mathcal{C}^{\infty}$ with heterogeneous amenities,
we can enumerate a ``large'' number of equilibria in a way qualified
later in the seciton.
\begin{thm}
\emph{(}\textbf{\emph{Homotopy $H_A$ from the City with Homogeneous Locations)}} There exists a differentiable
homotopy between the city $\mathcal{C}^{\infty}$ and the corresponding
city $\mathcal{C}_{h}^{\infty}$ with equal amenities across locations
$A_{1}=A_{2}=\dots=A_{J}$, equal marginal costs $\text{mc}=1$ and
equal weights across locations. 
\end{thm}

\begin{proof}
Start with an equilibrium vector $(z_{j})$ of the city with homogeneous
amenities as in proposition (\ref{prop:The-generalized-equilibria}),
define the following homotopy towards the city with heterogeneous
amenities, distances, and marginal costs:
\begin{align*}
A_{j}(t) & =(1-t)+tA_{j}\\
mc(t) & =(1-t)+tmc\\
\Delta(t) & =(1-t)\left[\begin{array}{ccc}
e & \cdots & e\\
\vdots &  & \vdots\\
e & \cdots & e
\end{array}\right]+t\Delta
\end{align*}
And $z_{j}(t)$ is the solution to:
\begin{equation}
P(\mathbf{z},t)=\sum_{l=1}^{J}A_{l}(t)mc_{l}(t)^{-\alpha}z_{l}(t)^{p}z_{j}(t)^{q}-\sum_{k=1}^{J}\Delta_{jk}(t)A_{k}(t)mc_{k}(t)^{-\alpha}z_{k}(t)^{p}=0\label{eq:equilibrium_homogeneous_amenities}
\end{equation}
Starting from a solution $\mathbf{z}^{*}$ of (\ref{eq:equilibrium_homogeneous_amenities})
at $t=0$, the solution $\mathbf{z}$ at $t=1$ to the differential
equation 
\[
\frac{\partial Q}{\partial\mathbf{z}}\frac{d\mathbf{z}}{dt}+\frac{\partial Q}{\partial t}=0,\qquad\mathbf{z}(0)=\mathbf{z}^{*}
\]
is an equilibrium of the economy with heterogenous amenities, marginal
cost, and the chosen weighting matrix $\Delta$ as long as $\partial Q/\partial\mathbf{z}$
is invertible.
\end{proof}

\subsection{Equilibria of the City when Housing Supply has Finite Supply Elasticity\label{sec:Finite-Housing-Supply-Elasticity}}

We now consider an economy with a finite housing supply elasticity
$0<\eta<\infty$. We can find the equilibria of this economy by starting
with an arbitrarily chosen solution $\mathbf{z}^{*}(\eta=0)\in\mathbb{R}^{J}$
of the perfectly elastic city, $\eta=+\infty$, and then smoothly
transforming the equilibria of this economy to the desired elasticity
$\eta=\eta^{*}$. It is possible to follow a path of $\mathbf{z}^{*}(\eta)$
in an open set around an initial solution $\mathbf{z}^{*}(0)$ when the Jacobian is invertible as an application of
the Implicit Function Theorem. 

We perform a change of variable by expressing the set of equilibria
as a function of the inverse of the elasticity $\zeta=1/\eta$ to
obtain a path on the compact set $[0,1/\eta]$. Also express $\mathbf{z}$
in logs to avoid issues related to the domain of $\mathbf{z}$. When
the path of solutions $\log\mathbf{z}(\zeta)$ is regular, an equilibrium
of the city with elastic housing supply $\zeta=1/\eta$ is found by
integrating $d\log\mathbf{z}(\zeta)\in\mathbb{R}^{J}$ from $\zeta=0$
to $\zeta=1/\eta^{*}$. 

To transform solutions smoothly from $\zeta$ to $\zeta+d\zeta$,
write the two sets of $J$ equilibrium conditions that determine the
equilibrium at a given $\zeta>0$ and take the differentials. 

\subsubsection{Equilibrium Conditions}

First, $J$ market clearing prices that ensure the equality of the
demand and supply of housing in each location:
\begin{equation}
q_{j}=\left(\text{mc}_{j}\right)^{\frac{\eta}{\alpha+\eta}}\left[\frac{L_{1}}{c_{j}}\frac{A_{j}\Psi_{j}^{\gamma^{1}}}{\sum_{k}A_{k}q_{k}^{-\alpha}\Psi_{k}^{\gamma^{1}}}+\frac{L_{2}}{c_{j}}\frac{A_{j}\Psi_{j}^{\gamma^{2}}}{\sum_{k}A_{k}q_{k}^{-\alpha}\Psi_{k}^{\gamma^{2}}}\right]^{\frac{1}{\alpha+\eta}}\label{eq:market_clearing}
\end{equation}
and we see that $q_{j}\rightarrow\text{mc}_{j}$ when $\eta\rightarrow+\infty$.
Define $\zeta\equiv1/\eta$. This becomes:
\begin{equation}
q_{j}=\left(\text{mc}_{j}\right)^{\frac{1}{\alpha\zeta+1}}\left[\frac{L_{1}}{c_{j}}\frac{A_{j}\Psi_{j}^{\gamma^{1}}}{\sum_{k}A_{k}q_{k}^{-\alpha}\Psi_{k}^{\gamma^{1}}}+\frac{L_{2}}{c_{j}}\frac{A_{j}\Psi_{j}^{\gamma^{2}}}{\sum_{k}A_{k}q_{k}^{-\alpha}\Psi_{k}^{\gamma^{2}}}\right]^{\zeta\frac{1}{\alpha\zeta+1}}\label{eq:price_market_determined}
\end{equation}
Use the change of variables of the previous section $\Psi_{j}^{1/q}=z_{j}$,
so that $\Psi_{j}^{\gamma^{1}}=z_{j}^{p}$ and $\Psi_{j}^{\gamma^{2}}=z_{j}^{q\gamma^{2}}$.
This change of variables is not crucial here: it merely ensures that
we can use as a starting point the solution $\mathbf{z}^{*}(\zeta=0)\in\mathbb{R}^{J}$
found in the previous section. To make sure that solutions $z_{j}$
to the system are always positive, express the system in terms of
$\log z_{j}$. Given that the derivative of $q_{j}^{-\alpha}$ is
better expressed as $-\alpha\log q_{j}$, also write the system in
$\log q_{j}$. Write each element $\log z_{j}$ of the vector $\log\boldsymbol{z}$
as a function of the inverse elasticity $\zeta\in\mathbb{R}^{+}$,
so that we can transform the solutions smoothly from $\zeta=0$ ($\eta=\infty$)
to $\zeta>0$ ($\eta\in\mathbb{R}$).

We obtain a first set of $J$ non-linear equations in $2J+1$ parameters.
Log prices are a function of log $\mathbf{z}$, log prices, and the
inverse elasticity:
\begin{equation}
Q(\log\mathbf{z},\log\mathbf{q};\zeta) = \log\mathbf{q} \label{eq:equilibrium_Q}
\end{equation}

The second set of equilibrium relationships is obtained as in the
previous section. The $J$ polynomials in $J$ unknowns $z_{j}$ ensure
that the social equilibrium conditions are satisfied.
\begin{equation}
z_{j}(\zeta)^{q}\sum_{l=1}^{J}A_{l}q_{l}(\zeta)^{-\alpha}z_{l}(\zeta)^{p}-\sum_{k=1}^{J}e^{-\xi d_{jk}} A_{k}q_{k}^{-\alpha}(\zeta)\left(z_{k}(\zeta)\right)^{p}=0
\end{equation}
This is a second set of $J$ equations in $2J+1$ parameters:
\begin{equation}
F(\log\mathbf{z},\log\mathbf{q};\zeta)=0
\end{equation}
The equilibrium of the economy is thus implicitly defined by a system
of $2J$ non-linear equations:
\begin{equation}
\left\{
\begin{array}{ll}
Q(\log\mathbf{z},\log\mathbf{q};\zeta) & =\log\mathbf{q}\\
F(\log\mathbf{z},\log\mathbf{q};\zeta) & =0
\end{array}
\right. \label{eq:system_eta_finite}
\end{equation}
where the inverse elasticity $\zeta$ varies in $(0,1/\eta]$. 

\subsubsection{Homotopy $H_\eta$ Along the Inverse of Supply Elasticity\label{subsec:Homotopy}}

Section~\ref{subsec:Perfectly-Elastic-Housing-Supply-Exact-Method}
has solved for $\mathbf{z}$ for $\zeta=0$ i.e. for $\eta=\infty$.
The key observation here is that the system provides a linear relationship
between $\log\mathbf{z}(\zeta)$ and the set of possible $\log\mathbf{z}(\zeta+d\zeta)$.
When the relationship is unique, i.e. the Jacobian is invertible:
\begin{equation}
\mathbf{z}(\zeta)=\mathbf{z}(0)e^{\int_{0}^{\zeta}\frac{d\log z}{d\zeta}d\zeta}\label{eq:integration_of_z}
\end{equation}
We obtain $\frac{d\log\mathbf{z}}{d\zeta}$ by taking the total derivative
of the equilibrium system of $2J$ equations w.r.t. $\zeta$. For
a given value of $\zeta$, take the Jacobians and gradients of $Q$
and $F$ w.r.t $\log\mathbf{z}$, $\log\mathbf{q}$, and $\zeta$:
\begin{align}
\frac{\partial Q}{\partial\log\mathbf{z}}\frac{d\log\mathbf{z}}{d\zeta}+\frac{\partial Q}{\partial\log\mathbf{q}}\frac{d\log\mathbf{q}}{d\zeta}+\frac{\partial Q}{\partial\zeta} & =\frac{\partial\log\mathbf{q}}{\partial\zeta}\label{eq:jacobian_system}\\
\frac{\partial F}{\partial\log\mathbf{z}}\frac{d\log\mathbf{z}}{d\zeta}+\frac{\partial F}{\partial\log\mathbf{q}}\frac{d\log\mathbf{q}}{d\zeta}+\frac{\partial F}{\partial\zeta} & =0\nonumber 
\end{align}
This is a differential equation in $\log \mathbf{y}=(\log\mathbf{z},\log\mathbf{q})$
with initial conditions at $\zeta=0$ the equilibrium found in the
previous section for the perfectly elastic city. We can simply write
this differential equation as:
\begin{equation}
\Gamma\frac{d\log \mathbf{y}}{d\zeta}=-\left( \begin{array}{c} \frac{\partial Q}{\partial\zeta} \\ \frac{\partial F}{\partial\zeta} \end{array} \right),\qquad \log \mathbf{y(0)}=(\log\mathbf{z}(0),\log\mathbf{mc})\label{eq:differential_equation}
\end{equation}
We get:
\begin{prop}
\textbf{\emph{(Equilibria of the City with Elastic Housing Supply using Homotopy $H_\eta$)}}\label{prop:homotopy_H_eta}
Consider an equilibrium $(\mathbf{x}^{*},\mathbf{q}^{*})$ of the
city with perfectly elastic housing supply. Consider a solution $\log \mathbf{y}(\zeta)=(\log\mathbf{z}(\zeta),\log\mathbf{q}(\zeta))$
to the differential equation with initial condition $(\log\mathbf{z}(0),\log\mathbf{q}(0))=(\log(\Delta\mathbf{x}^{*})^{1/q},\log\mathbf{mc})$ an equilibrium of the city with \emph{perfectly} elastic housing supply.
Then the pair $(\mathbf{x}^{*}(\zeta),\mathbf{q}^{*}(\zeta))=(\Delta^{-1}\mathbf{z}(\zeta)^{q},\mathbf{q}(\zeta))$
is an equilibrium of the city with \emph{elastic} housing supply.

In the case when $\left(1-\frac{\partial Q}{\partial\log\mathbf{q}}\right)$
and $\left[\frac{\partial F}{\partial\log\mathbf{z}}+\frac{\partial F}{\partial\log\mathbf{q}}\left(1-\frac{\partial Q}{\partial\log\mathbf{q}}\right)^{-1}\frac{\partial Q}{\partial\log\mathbf{z}}\right]$
are invertible,\footnote{This is related to the formula for the inverse of the partitioned matrix. The second quantity is called the Schur complement \cite{zhang2006schur} of the block  $\left(1-\frac{\partial Q}{\partial\log\mathbf{q}}\right)$ of the Jacobian.} this yields the unique integrated solution:
\begin{equation}
\mathbf{z}(\zeta)=\mathbf{z}(0)\exp^{-\int_{0}^{\zeta}\left[\frac{\partial F}{\partial\log\mathbf{z}}+\frac{\partial F}{\partial\log\mathbf{q}}\left(1-\frac{\partial Q}{\partial\log\mathbf{q}}\right)^{-1}\frac{\partial Q}{\partial\log\mathbf{z}}\right]^{-1}\left[\frac{\partial F}{\partial\zeta}+\frac{\partial F}{\partial\log\mathbf{q}}\left(1-\frac{\partial Q}{\partial\log\mathbf{q}}\right)^{-1}\frac{\partial Q}{\partial\zeta}\right]ds}\label{eq:PDE_z}
\end{equation}
When one of the two matrix inverses are not unique, differential equation (\ref{eq:differential_equation}) may have multiple
solutions. This is done in the discussion and extension section. 
\end{prop}

To check for the existence of bifurcation points, we compute the smallest
absolute value of the eigenvalues of each matrix as well as the condition number
 along the path $\zeta\in[0,\frac{1}{\eta}]$.\footnote{The empirical application presents the case of an economy with a singular
point, where the smallest eigenvalue gets closer to zero as $\zeta$
increases.}

In practice, the closed form derivation of Jacobians and gradients
in the system in $\mathbf{z}$ and $\mathbf{q}$ (\ref{eq:jacobian_system}) may be time consuming.
We use forward mode automatic differentiation to obtain these Jacobians
and gradients~\cite{wengert1964simple}. The market clearing conditions $Q$ and their partial derivatives $\partial Q/\partial \log \mathbf{q}$ have well-defined limits as $\zeta\rightarrow 0$ from above.  

\bigskip{}

Given that the set $\mathcal{Z}(0)$ of solutions $\mathbf{z}(0)$
is known from the previous procedure, this procedure can be applied
to any solution in $\mathcal{Z}(0)$ to generate a set $\mathcal{Z}(\zeta)$
of solutions of the elastic city. Notice that some initial conditions
$\mathbf{z}(0)$ with $\mathbf{x}=\Delta^{-1}\mathbf{z}(0)^{q}\notin[0,1]^{J}$
may not be equilibria of the economy with perfectly elastic housing
supply but may lead to equilibria for the elastic city $\mathbf{z}(\zeta)$. 

\bigskip{}

\subsection{Recap: Practical Implementation}

This provides us with a three-step approach to the computation of
equilibria:
\begin{itemize}
\item In the first step, calculate the exact set of solutions $\mathbf{z}(0)\in\mathbb{\mathbb{C}}^{J}$
for the system of polynomials~(\ref{eq:polynomial_system}). Denote
this set by $\mathcal{Z}(0)$.
\item In the second step, for each element $\mathbf{z}(0)\in\mathcal{Z}(0)$,
solve the differential equation~(\ref{eq:PDE_z}) to find $z(\zeta)$
numerically using two-step Runge-Kutta or other predictor-corrector
methods. Intuitively, the solution $\mathbf{z}(0)$ can be ``gradually
adjusted'' to reach $\mathbf{z}(\zeta)$ for each solution $\mathbf{z}(0)$. At each point along the path, compute the rank of the Jacobian of the homotopy $H$ to detect singularities.
\item Finally, transform the equilibria $\mathbf{z}(\eta)$ into $\Psi(\eta)=\mathbf{z}(\eta)^{q}$
and then the equilibrium population of type 1 as $\mathbf{x}(\eta)=\Delta^{-1}\Psi(\eta)$. Keep those proper equilibria, i.e. for which $\mathbf{x}(\eta)$ are strictly positive.
\end{itemize}

\section{Examples}\label{sec:examples}

We present examples using synthetic data before proceeding in the
next section with analysis on historical data. Cities differ in the
following dimensions. Amenities are log-normally distributed $\log(A)\sim N(\mu,\sigma^{2})$.

\begin{table}[t]
\caption{Parameters of the City}

\bigskip

\emph{This table describes the parameters used in the examples of Section~\ref{sec:examples}. These parameters are those of the city $\mathcal{C}$, introduced in Definition~\ref{def:(City-Equilibrium)-An}, where the amenity vector $\mathbf{A}$ is drawn from a log normal distribution with mean $\mu$ and variance $\sigma^2$.}

\bigskip

\centering{}%
\begin{tabular}{lc}
\textbf{Parameter} & \textbf{Symbol}\tabularnewline
\hline 
Number of neighborhoods & $J$\tabularnewline
Group 1 residents' preference for group 1 neighbors & $\gamma^{1}$\tabularnewline
Group 2 residents' preference for group 2 neighbors & $\gamma^{2}$\tabularnewline
Mean of log normal Amenities & $\mu$\tabularnewline
Standard deviation of log normal Amenities & $\sigma$\tabularnewline
Marginal cost of floor surface & $mc$\tabularnewline
Total population & $L$\tabularnewline
Total population of group $g$ & $L_{g}$\tabularnewline
Share of housing in budget (Elasticity of demand) & $\alpha$\tabularnewline
Scope of social interactions & $\xi$\tabularnewline
Supply elasticity of floor surface & $\eta$\tabularnewline
Constant of the supply curve & $c$\tabularnewline
\end{tabular}
\end{table}

Across examples we vary five parameters the number of neighborhoods
$J$, households' preferences ($\gamma^{1}$, $\gamma^{2}$), the
standard deviation of the log distribution of amenities $\sigma$,
and the scope $\xi$ of social interactions. We keep constant the
mean of log normal amenities to 0, the marginal cost of floor surface
to $mc=1$, non-college educated population to 80\% of $J$, college
educated population to 20\% of $J$, the share of housing in the budget
to $\alpha=0.3$, the elasticity of floor surface supply is either
$\infty$ (perfectly elastic) or 0.8 (Elasticity of supply for San
Francisco from \citeasnoun{saiz2010geographic}), and the constant of the supply curve is $c_j=1$.

\subsection{Equilibria with $J$ locations and Social Interactions}

Table~\ref{tab:Finding-Equilibria-by-Method-1} presents the equilibria
found for cities with heterogeneous amenities and social preferences
but a constant price. The next section presents results with $\eta<\infty$. 

We estimate the equilibria for 2,220 cities with a range of structural parameters and estimate the number and location
of equilibria for each location. Results are available in the data archive. Table~\ref{tab:Finding-Equilibria-by-Method-1} presents a subset of 15 cities, with social preferences $\gamma^1$ ranging between 0.2 and 5.0, the scope of social interactions varies between 0.0 and 4.0, the standard deviation of log amenities between 0.1 and 1.5. Overall, in the table, the number of equilibria can be as high as 127 and as small as 1. Unsurprisingly, this obtains in the case of a larger number of locations, for stronger social preferences ($\gamma^1=4.0$) and for a low dispersion of amenities; in other words there is a large number of equilibria in cases where the geography does not pin down the equilibrium, in the spirit of \citeasnoun{beckmann1987location}.

\begin{figure}

\caption{Number of Equilibria and the City's Parameters}

\footnotesize
\emph{These plots present the outcome of the determination of multiple equilibria for a set of 2,223 cities with a range of structural parameters. The scope $\xi$ measures how much residents care about neighboring demographics: a high $\xi$ suggests that residents weigh their close neighbors more heavily than their distant neighbors. A low $\xi$ suggests that residents care about neighbors far and close. These estimations are done for the city with perfectly elastic supply of floor surface.}

\bigskip

\centering
\begin{subfigure}[b]{0.45\textwidth}
\centering
\caption{\footnotesize Scope $\xi$}
\includegraphics[scale=0.35]{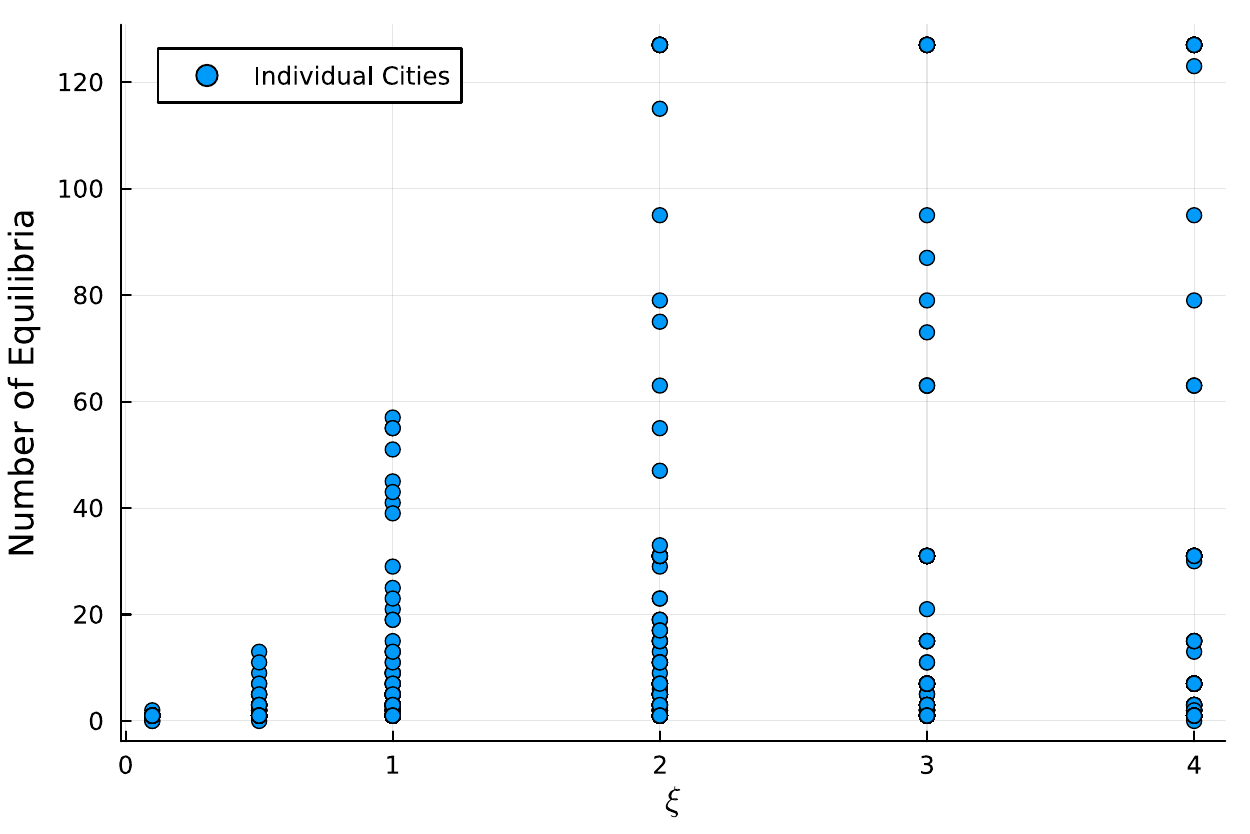}    
\end{subfigure}
\hfill
\begin{subfigure}[b]{0.45\textwidth}
\centering
\caption{\footnotesize Social Preferences $\gamma^1$}
\includegraphics[scale=0.35]{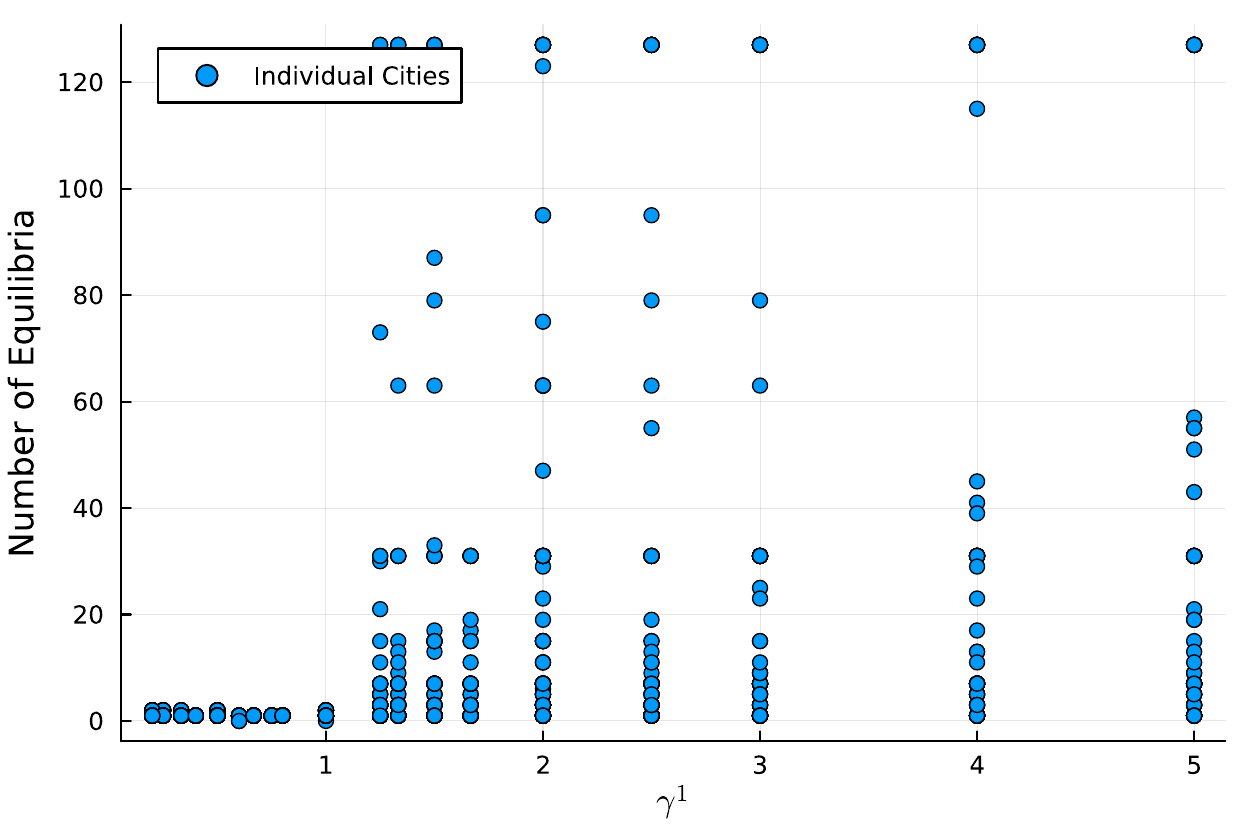}    
\end{subfigure}

\end{figure}

\subsection{Equilibria with Finite Elasticity of Housing Supply}

For each of the 2,223 cities of the previous section, we then perform a homotopy continuation 
from the infinitely elastic ($\eta=\infty$) city to the elastic city ($\eta<\infty$).
We use a range of $\eta\in(0,0.67]$, where 0.67 is the elasticity of supply for the MSA of San
Francisco in \cite{saiz2010geographic}, and we include $\eta$s close to the perfectly
inelastic city $\eta\rightarrow 0$.

Price responses are pecuniary externalities that are intuitively a source of \emph{strategic substitutability} (e.g. as in Cournot games) as opposed to social preferences $\gamma^1$, which are a source of strategic complementarities (e.g. in coordination games). We thus  expect to observe fewer equilibria as we constrain the supply of space in each location by lowering $\eta$. 

This is what is displayed on Figure~\ref{fig:systematic_finite_elasticity}. The vertical axis is the average number of equilibria for each value of $\gamma^w=p/q$. This figure synthesizes the results of the homotopy $H_\eta$, where each starting is a city with heterogeneous amenities and perfectly elastic housing supply. 

\begin{figure}
    \caption{Systematic Analysis of the 2,223 Cities, from the Perfectly Elastic Supply Case to Finite Elasticity}
    \emph{This plot presents the average number of equilibria for different values of the preference for white neighbors ($\gamma^w=p/q$), in the perfectly elastic case ($\eta = \infty$) and in the case of finite elasticity. In this latter case $\eta=0.67$, which is the supply elasticity of the San Francisco Bay area in \citeasnoun{saiz2010geographic}.}
    \bigskip
    \label{fig:systematic_finite_elasticity}
    \begin{center}
    \includegraphics[scale=0.4]{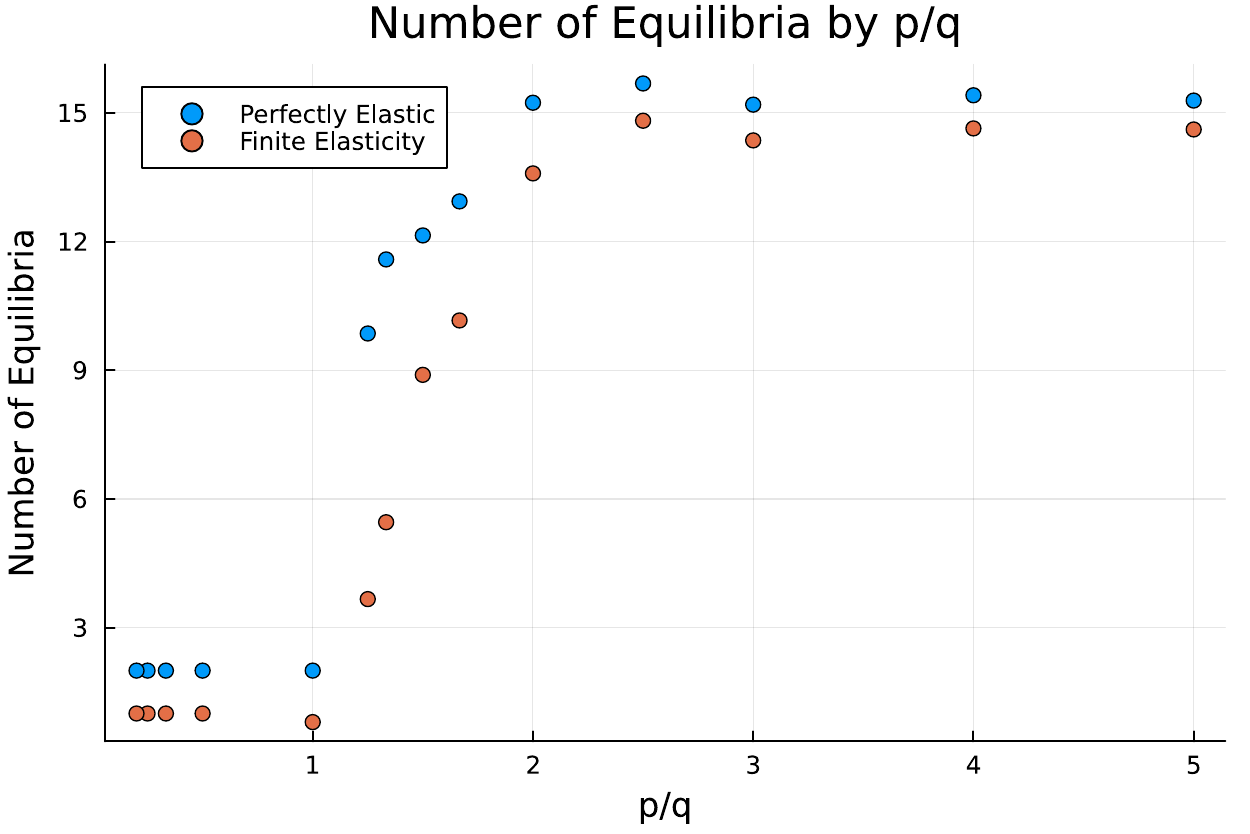}
    \end{center}
\end{figure}

\clearpage 
\pagebreak

\section{Larger Cities with an Application to Chicago}\label{sec:empirics}

The approaches presented in the paper guarantee the enumeration of equilibria yet may be computationally challenging. 

Here we develop a tractable approach for large cities, using the concepts of \emph{near} and \emph{far} interactions. Households care both about the characteristics of their neighbors in the immediate vicinity, but also about the characteristics of households in nearby communities. We apply this approach to the case of the City of Chicago's spatial equilibrium in racial composition across communities using a nested choice approach which is computationally tractable and enumerates the city's equilibria. Consider the choice of households among 353 neighborhoods grouped in 77 typical communities of Chicago, including Humboldt Park, Irving Park, the Loop, South Chicago, and 9 regions, as depicted on Figure~\ref{fig:chicago_map}.

\begin{figure}
    \caption{Chicago's 77 Communities and their Neighborhoods}
    \label{fig:chicago_map}
    \begin{center}
    \includegraphics{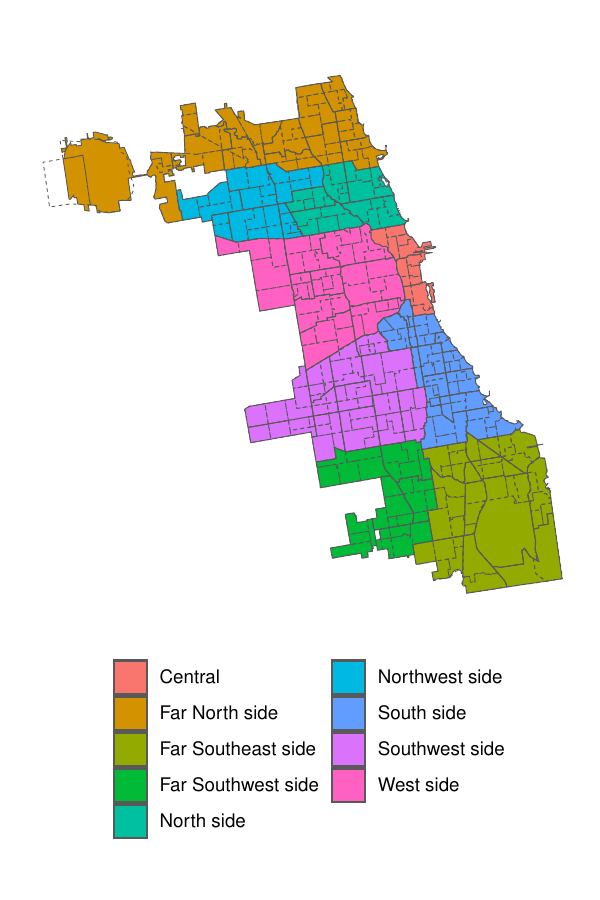}
    \end{center}
    \emph{Solid line: 77 community area boundaries, designed by the Social Science Research Committee at the University of Chicago. Dotted lines: Neighborhoods from 2010 Census tracts. Colored polygons: Regions.}
\end{figure}

\begin{figure}
    \caption{Estimation of the Model's Structural Parameters}
    \label{fig:scope_estimation}
    
    \emph{Estimation of amenities $A_i$ and social preferences $\gamma^w$, $\gamma^b$ using a regression of log population on log nearby neighborhood demographics, Panel (b). The scope parameter $\xi$, also called rate of decay in \citeasnoun{redding2017quantitative} is the maximum of the fit statistic in Panel (a).}

\bigskip
    
    \begin{center}
    (a)~$R^2$ against Scope Parameter $\xi$
    
    \bigskip
    
    \includegraphics[scale=0.5]{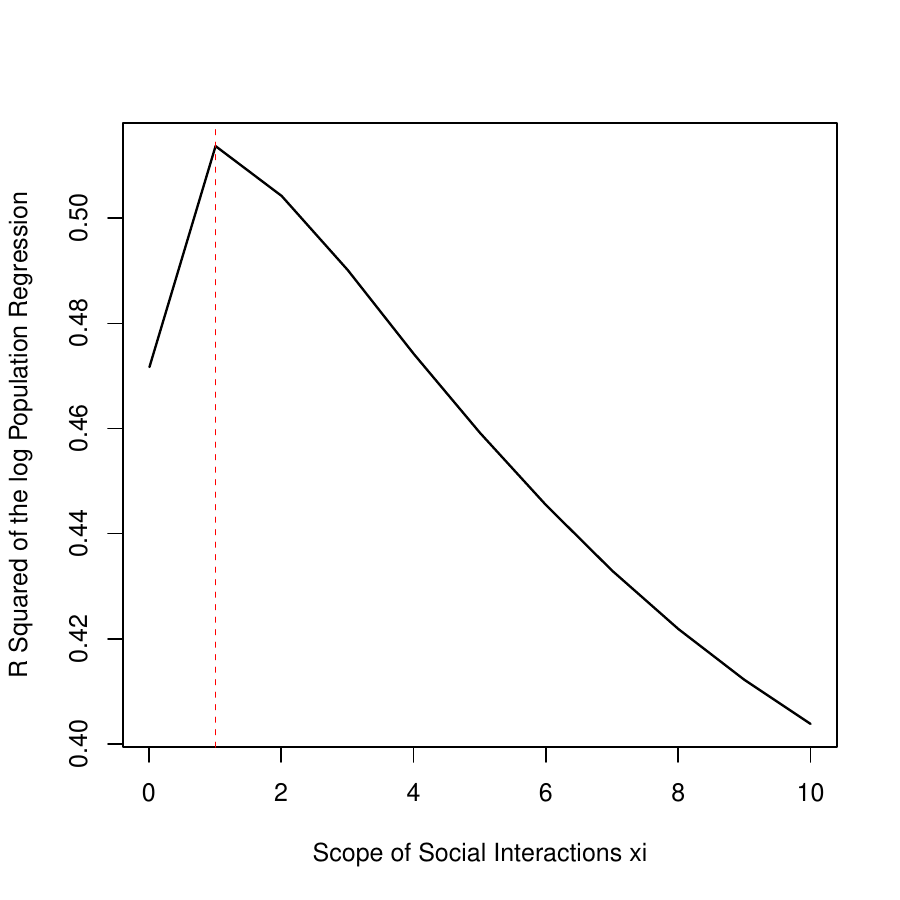}
    \end{center}

    \begin{center}
    (b)~Estimation of Amenities and Social Preferences $\gamma^b$, $\gamma^w$
    
    \bigskip
    
    \begin{tabular}{lc}
   \toprule
                                & $\log N^r_{jt}$\\  
   \midrule
   $r=$Black  $\times\log(\Psi_{jt-1})$, coefficient $\gamma^b$     & $-$0.4391$^{**}$\\   
                                              & (0.1375)\\   
   $r=$White $\times\log(\Psi_{jt-1})$, coefficient $\gamma^w$    & $+$2.003$^{***}$\\   
                                              & (0.2859)\\   
   \midrule
   Neighborhood  Fixed-effects  $j$                           & Yes\\  
   Year          Fixed-effects  $t$                           & Yes\\  
   \midrule
   Observations (Neighborhood $\times$ Year) & 4,354\\  
   R$^2$                                      & 0.51371\\  
   Within R$^2$                               & 0.40499\\  
   \bottomrule
   \multicolumn{2}{l}{\emph{Double-Clustered (Neighborhood $\times$ Year) standard-errors in parentheses}}\\
   \multicolumn{2}{l}{\emph{Signif. Codes: ***: 0.01, **: 0.05, *: 0.1}}\\
\end{tabular}

    \end{center}
\end{figure}

\begin{figure}
    \caption{Multiple Equilibria of Humboldt Park (Homotopy $H_A$, perfectly elastic housing supply)}
    \label{fig:equilibria_one_community}
    \begin{center}
    \includegraphics[scale=0.25]{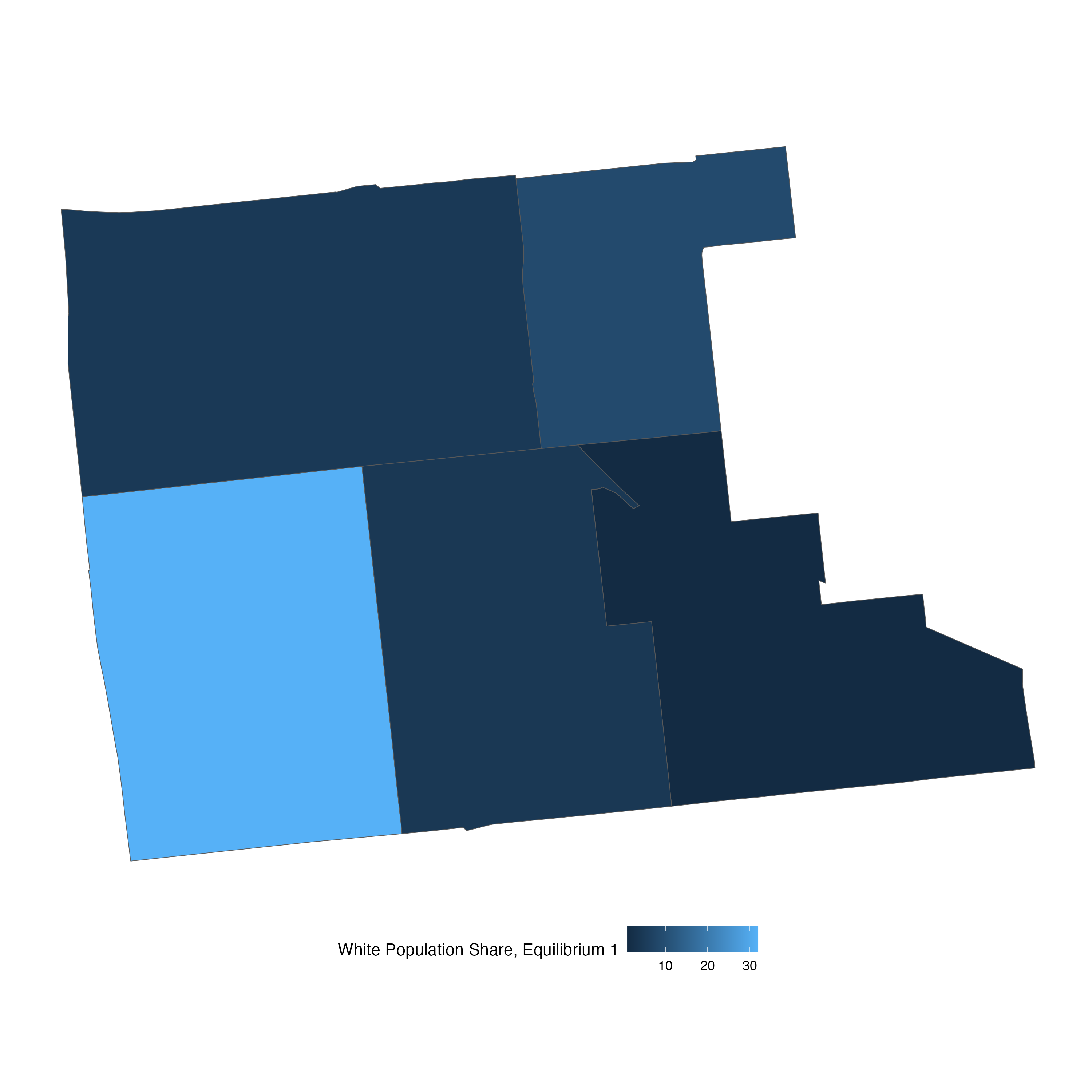}
    \includegraphics[scale=0.25]{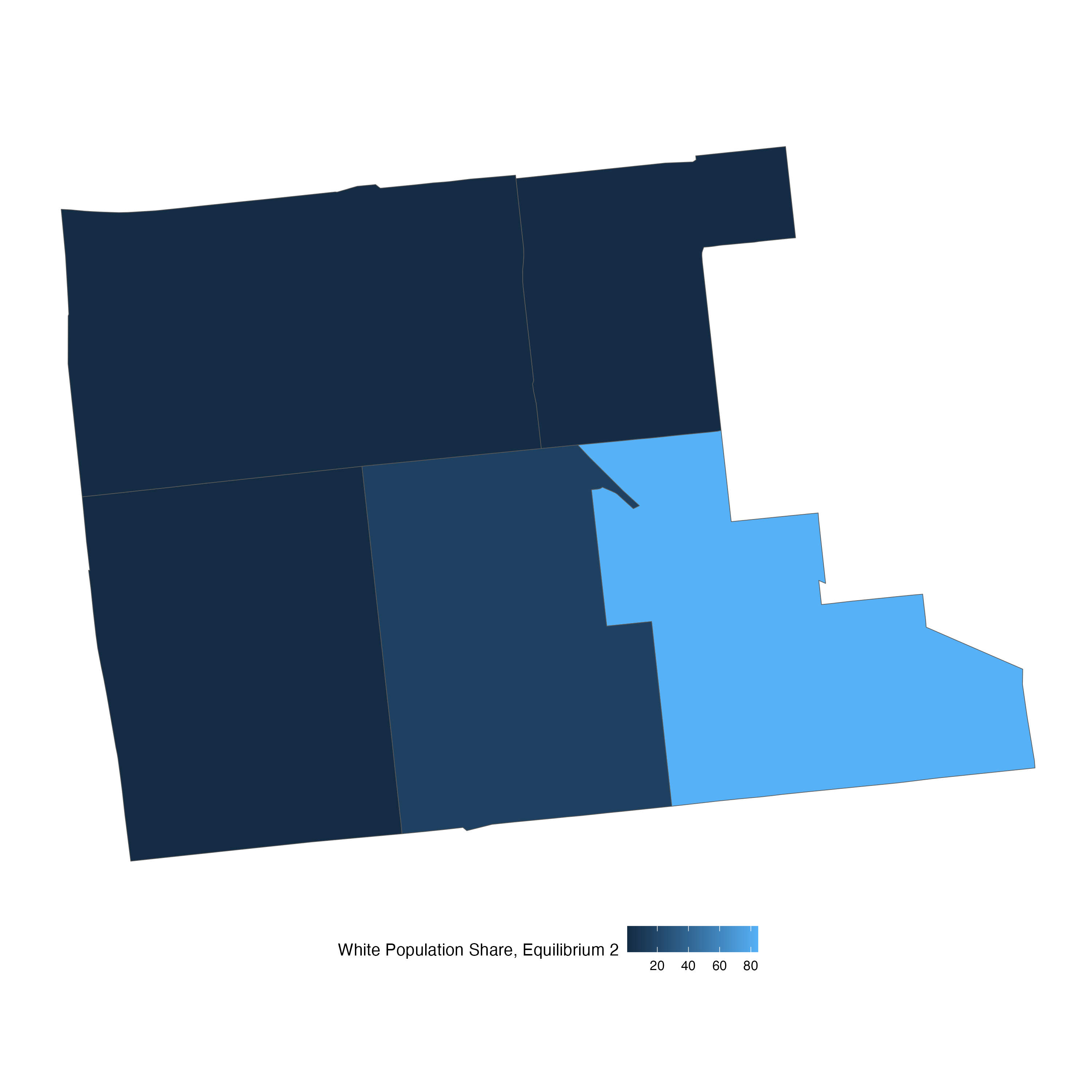}
    \includegraphics[scale=0.25]{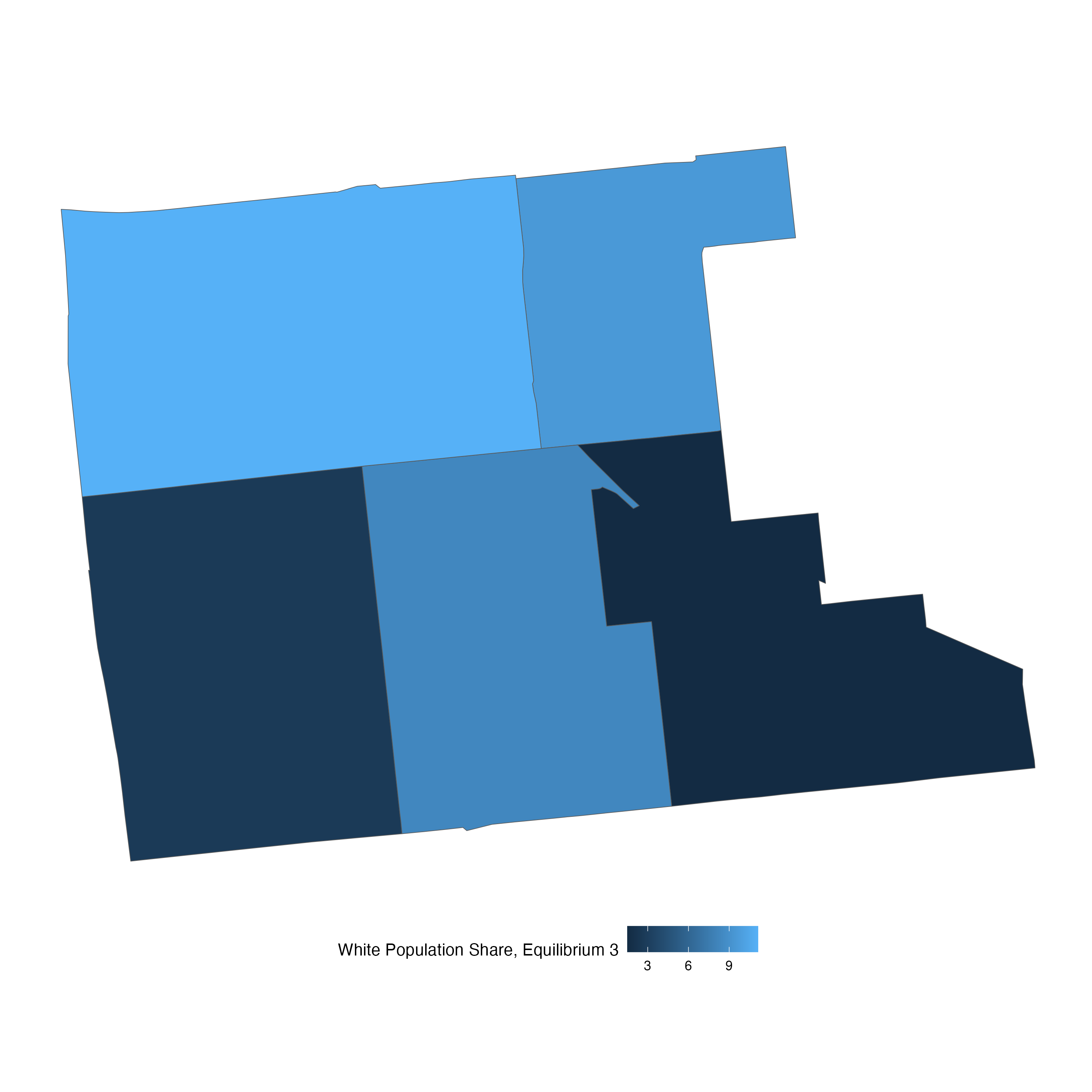}
    \includegraphics[scale=0.25]{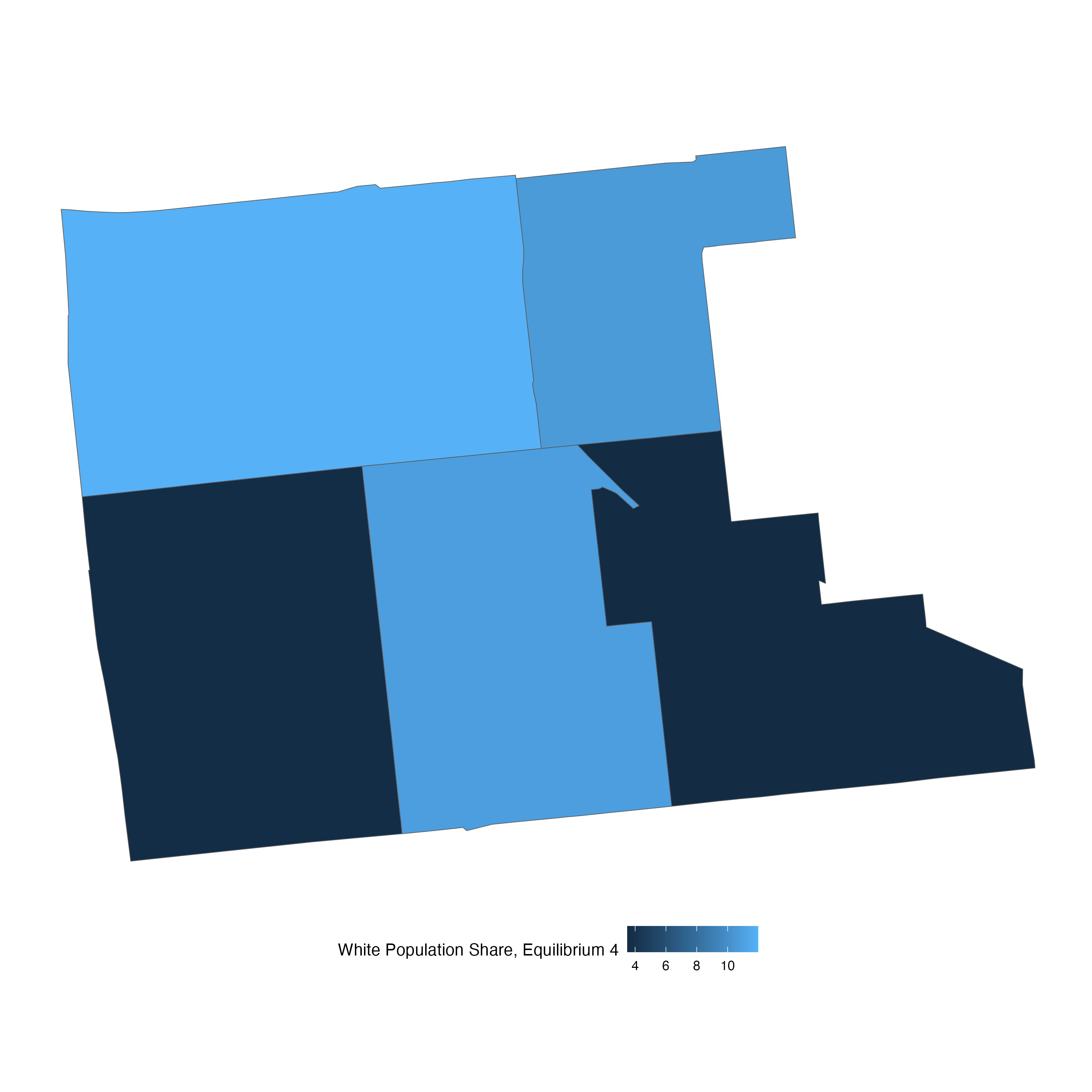}
    \includegraphics[scale=0.25]{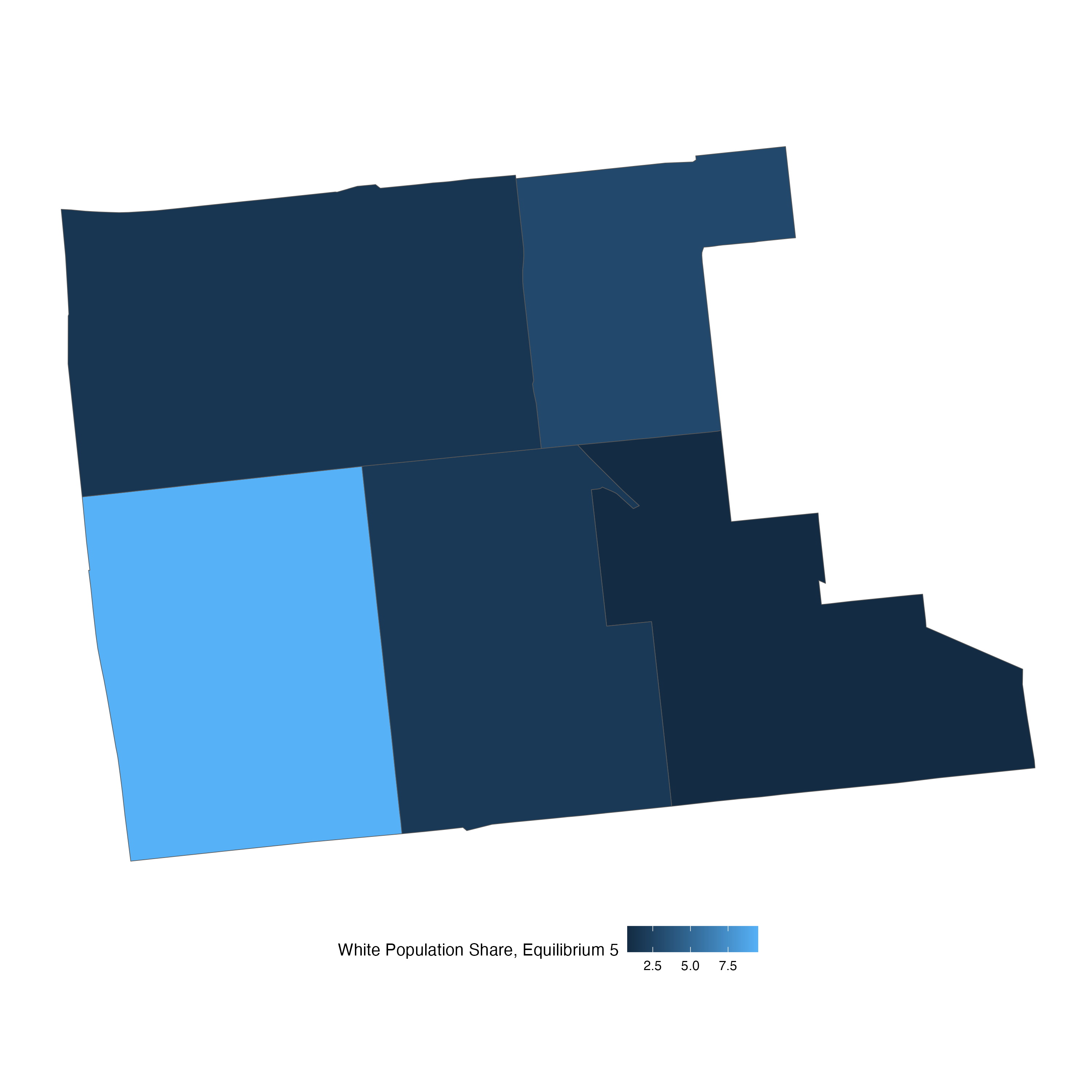}
    \end{center}
    \emph{The five equilibria of Humboldt Park, assuming elastic supply of floor surface, obtained by the homotopy $H_A$.}
\end{figure}

\begin{table}

\caption{Number of Equilibria Estimated within Each of Chicago's 77 Communities (Homotopy $H_A$, perfectly elastic housing supply)}\label{tab:number_equilibria_community}

\emph{For each of Chicago's 77 communities, this table presents the equilibria estimated using homotopy $H_A$, and the estimated parameters with the longitudinal panel 1950-2010. $\dagger$: the equilibria for Humboldt Park are depicted on Figure~\ref{fig:equilibria_one_community}.}

\begin{center}
{
\footnotesize
\begin{tabular}{lclc}
  \toprule
Community & Number of Equilibria & Community & Number of Equilibria \\ 
  \midrule
Albany Park & 3 & Lincoln Park & 1 \\ 
  Archer Heights & 1 & Lincoln Square & 3 \\ 
  Armour Square & 1 & Logan Square & 3 \\ 
  Ashburn & 5 & Loop & 1 \\ 
  Auburn Gresham & 1 & Lower West Side & 1 \\ 
  Austin & 9 & Mckinley Park & 1 \\ 
  Avalon Park & 1 & Montclare & 1 \\ 
  Avondale & 1 & Morgan Park & 1 \\ 
  Belmont Cragin & 1 & Mount Greenwood & 3 \\ 
  Beverly & 5 & Near North Side & 3 \\ 
  Bridgeport & 1 & Near South Side & 1 \\ 
  Brighton Park & 3 & Near West Side & 3 \\ 
  Burnside & 1 & New City & 7 \\ 
  Calumet Heights & 1 & North Center & 1 \\ 
  Chatham & 1 & North Lawndale & 5 \\ 
  Chicago Lawn & 1 & North Park & 1 \\ 
  Clearing & 5 & Norwood Park & 5 \\ 
  Douglas & 1 & Oakland & 1 \\ 
  Dunning & 3 & Ohare & 7 \\ 
  East Garfield Park & 3 & Portage Park & 7 \\ 
  East Side & 1 & Pullman & 1 \\ 
  Edgewater & 1 & Riverdale & 3 \\ 
  Edison Park & 1 & Rogers Park & 1 \\ 
  Englewood & 3 & Roseland & 5 \\ 
  Forest Glen & 5 & South Chicago & 1 \\ 
  Fuller Park & 3 & South Deering & 7 \\ 
  Gage Park & 1 & South Lawndale & 3 \\ 
  Garfield Ridge & 7 & South Shore & 1 \\ 
  Grand Boulevard & 1 & Uptown & 1 \\ 
  Greater Grand Crossing & 3 & Washington Heights & 3 \\ 
  Hegewisch & 3 & Washington Park & 1 \\ 
  Hermosa & 1 & West Elsdon & 1 \\ 
  \textbf{Humboldt Park}$\dagger$ & 5 & West Englewood & 1 \\ 
  Hyde Park & 1 & West Garfield Park & 1 \\ 
  Irving Park & 3 & West Lawn & 1 \\ 
  Jefferson Park & 3 & West Pullman & 5 \\ 
  Kenwood & 1 & West Ridge & 3 \\ 
  Lake View & 1 & West Town & 3 \\ 
   &  & Woodlawn & 1 \\ 
   \bottomrule
\end{tabular}}
\end{center}

\end{table}

\subsection{Equilibrium Concept}

The city is made of neighborhoods indexed by $j$ grouped into communities indexed $i$. There are two populations $L^b$ and $L^w$ ($g=b,w$), with social preference parameters $\gamma^b$ and $\gamma^w$. There are $J_i$ neighborhoods within each community. The amenities are $A_{ij}$. The price of floor surface in community $i$ and neighborhood $j$ is $q_{ij}$, and the social demographics are $\Psi_{ij}$. Households choose a community $i$ and a neighborhood $j$. The utility of choosing neighborhood $j$ in community $i$ is:
\begin{equation}
    U_{ij}^g = A_{ij} q_{ij}^{-\alpha} \Psi_{ij}^{\gamma^g} \Psi_{i}^{\gamma^g}
\end{equation}
where $\Psi_{ij}$ capture the `\emph{near} interactions' across neighborhoods across communities (e.g. Humboldt Park) and $\Psi_{i}$ captures `far interactions' between communities (e.g. between Humboldt Park and West Town).
\begin{equation}
    \underbrace{\Psi_{ij} = \sum_{k=1}^{J_i} e^{-\xi d_{ijk}} \frac{L_{ik}^g}{L_i^g}}_{\textrm{Near interactions}}, \qquad \underbrace{\Psi_{i} = \sum_{\iota =1}^n e^{-\xi d_{i\iota}} \frac{L_{\iota}^g}{N}}_{\textrm{Far interactions}}
\end{equation}
The probability of choosing neighborhood $j$ conditional on choosing community $i$ is:
\begin{equation}
    P_{j\vert i} = \frac{\left( A_{ij} q_{ij}^{-\alpha} \Psi_{ij}^\gamma \right)^\theta }{\sum_{j=1}^{N_i} \left( A_{ij} q_{ij}^{-\alpha} \Psi_{ij}^\gamma \right)^{\theta} }
\end{equation}
which leads to a social equilibrium condition within each community. The market clearing condition holds with floor surface elasticity $\eta$ and pins down $q_{ij}$.

Denoting by $V_i(L_i)$ the welfare in community $i$ conditional on $L_i$, this is equal to the typical:
\begin{equation}
V_i = \Gamma \left[ \sum_{j=1}^{J_i} \left( A_{ij} q_{ij}^{-\alpha} \Psi_{ij}^\gamma \Psi_{i}^\gamma \right)^\theta  \right]^{\frac{1}{\theta}} 
\end{equation}
where, as usual, $\Gamma$ is Euler's gamma function. This welfare can be factored into:
\begin{equation}
  V_i(L_i, \Psi_i) = \Psi_i^\gamma U_i(L_i),
\end{equation} 
hence  far interactions can be factored out at the community level. The probability of choosing community $i$ is then driven by the
\begin{equation}
    P_i = \frac{\Psi_i^{\gamma\theta} U_i(L_i)^\theta}{\sum_{k=1}^N \Psi_k^{\gamma\theta} U_k(L_k)^\theta }
\end{equation}
The Fr\'echet unobservable in neighborhood $j$ within community $i$ is observed after the choice of community~$i$.

\subsection{Equilibrium Solution with Perfectly Elastic Floor Supply Elasticity}

When floor surface is perfectly elastic ($\eta=\infty$), each community's welfare is scale independent, i.e. it is not the total population in the community that affects equilibria but the distribution of populations across locations. The probability of choosing community $i$ can be further simplified as:
\begin{equation}
    P_i = \frac{\Psi_i^{\gamma\theta} U_i(1)^\theta}{\sum_{k=1}^N \Psi_k^{\gamma\theta} U_k(1)^\theta }
\end{equation}
and the problems of choosing a community and a neighborhood within the community are nested and separately solvable, while allowing for both far and near interactions. 

In a second step, to allow for a finite elasticity of floor surface supply, we apply the second homotopy $H_\eta$ only once the equilibria for the entire city of Chicago have been found. The second homotopy is a differential equation over the number of locations, and its complexity is linear in the number of locations. It is therefore not a computational challenge to transform the city differentially from $\eta = \infty$ to a finite value $\eta \geq 0$.

The equilibria of the city with perfectly elastic housing supply are solved as follows. First, for each given level of population $L_i$ in a community, the methods presented in this paper can be applied to estimate the equilibria in $(L_{i1}, \dots, L_{iN_i})$ \emph{within the community}. Second, this provides us with a set of valuations $V_i(L_i) \subset \mathbb{R}^{N_i}$ for each level of population $L_i$ in the region. Equilibrium multiplicity implies that the valuation $V_i(L_i)$ is a set of cardinality $\geq 1$: there are multiple population distributions consistent with the structural parameters of amenities and social preferences. In mathematical terms, $V_i(L_i)$ is set valued. The proportionality described in the previous paragraph implies that $V_i(L_i) = L_i^{\gamma^1} V_i(1) \subset \mathbb{R}^{N_i}$. We then obtain the city's equilibria by solving for:
\begin{equation}
    \Psi_i = \sum_{\iota=1}^n e^{-\xi d_{i\iota}} \frac{\Psi_\iota^{\gamma\theta} U_\iota(1)^\theta}{\sum_{k=1}^N \Psi_k^{\gamma\theta} U_k(1)^\theta }  \label{eq:fixed_point}
\end{equation}
using the methods described in this paper. This is a problem of dimension the number of communities. 

Fix a choice of equilibrium in each community $i$. This is a choice of $U_i(1)$ in each community. If the number of equilibria is $N^e_i$ in community, there are:
\begin{equation}
    N^e = \Pi_{i=1}^N N^e_i
\end{equation}
possible choices of equilibria. Once chosen, this equilibrium in each community means that welfare is pinned down in each community. This means that $U_i(1)$ is a scalar in $\mathbb{R}$. Thus, by picking an equilibrium in each community, we can then find the equilibria of the overall city, when floor surface supply is perfectly elastic.

We then apply the homotopy $H_\eta$ at the city level to obtain the equilibrium of the city with an arbitrary flood supply elasticity. This is an application of Proposition~\ref{prop:homotopy_H_eta} presented earlier in the paper.\footnote{A straightforward extension is to allow for a different elasticity in each location.} The method is described in more detail below.

\subsection{Equilibria when Floor Supply has Finite Elasticity: A Citywide Homotopy}

When floor supply is elastic, prices respond to the demand for each location and thus the welfare in each community is \emph{not} homogeneous in population size. It is however possible to start with a solution to the problem with perfectly elastic floor supply and then smoothly, using a homotopy, differentiate the solution. This provides a path of equilibria of the city for each $\eta$, with a starting point a solution to \ref{eq:fixed_point}.

Consider one equilibrium selected in each community $i$. This fixes $U_i(1;\eta=\infty)$, which is now made clearly dependent on $\eta$. Consider the vector of populations $L_i(\eta=\infty)$ thus found as a solution to \ref{eq:fixed_point}.  The differential of welfare in community $i$ wrt the floor surface elasticity $\eta$ is obtained using Proposition~\ref{prop:homotopy_H_eta}. What remains to be established is the differential equation that $L_i(\eta)$ satisifies, since the homogeneity result does not apply when prices respond to demand: a larger population in the community leads to higher floor surface prices. We thus need to rely on the differential of (\ref{eq:fixed_point}) wrt $L_i$. This will then allow us to obtain multiple paths of equilibria of the city with finite floor surface elasticity.

As before, consider the differential w.r.t. the inverse elasticity $\zeta$ rather than the elasticity $\eta$ itself. It was shown earlier that the homotopy is continuous at $\zeta \rightarrow 0$ ($\eta \rightarrow \infty$) and thus the homotopy can be extended by continuity at $\zeta = 0$. 

\paragraph*{The Homotopy $H_\eta$} To perform the homotopy by lowering the supply elasticity (and thus increasing the inverse elasticity),  we would in principle need to keep track of population levels by community $\mathbf{L}^w,\mathbf{L}^b$, the utility level in each community $\mathbf{U}^w,\mathbf{U}^b$, the social interactions parameter $\boldsymbol{\Psi}_c$ of far interactions within the region, the vector of prices $\mathbf{q}_n$ by neighborhood, and the social interactions parameter $\boldsymbol{\Psi}_n$ of near interactions, within each community. 

Such complexity is not required, and it is possible to keep track of the vector of neighborhood prices only. 

\begin{prop}{\emph{\textbf{(Citywide Homotopy)}}}
Denoting by $\mathbf{q}_n$ the vector of prices in each neighborhood, the homotopy is the solution to the ordinary differential equation:
\begin{equation}
    \left( 1- \Phi \right) \frac{d \mathbf{q}_n}{d \zeta} = B \label{eq:citywide_ode}
\end{equation}
where $A$ and $B$ are two matrices specified below, whose value changes along the path $\zeta$. When $\left(\mathbf{1} - A\right)$ is invertible, finding the equilibrium of the city is integrating $d\mathbf{q}/d\zeta$ from $\zeta = 0$ (perfectly elastic floor surface supply) to the desired level $\zeta$.
\begin{equation}
    \mathbf{q}_n(\zeta) = \mathbf{mc} + \int_0^\zeta \left(\mathbf{1} - \Phi(s) \right)^{-1} B(s) ds
\end{equation}
where the initial condition $\mathbf{q}_n(0) = \mathbf{mc}$ is the marginal cost $\mathbf{mc}$ of floor surface. The invertibility of $\left(\mathbf{1} - \Phi(s) \right)$ is tested along the path by computing the determinant and the smallest eigenvalue.

The endogenous equilibrium vectors $\left[ \mathbf{L}^w ~ \mathbf{L}^b ~ \mathbf{U}^w ~ \mathbf{U}^b ~ \boldsymbol{\Psi}_c ~  \boldsymbol{\Psi}_n ~ \mathbf{q}_n \right]$ are then pinned down by the solution $\mathbf{q}_n(\zeta)$.
\end{prop}
\begin{proof}
We then express the matrices $\Phi$ and $B$. We lay out the equilibrium conditions and then solve for $d\mathbf{q}/d\zeta$.

Population levels are changing because of social interactions and the welfare in each community:
\begin{eqnarray}
\frac{d\mathbf{L}^w}{d\zeta} & = & \frac{\partial \mathbf{L}^w}{\partial \boldsymbol{U}^w} \frac{d \boldsymbol{U}^w}{d\zeta} 
 + \frac{\partial \mathbf{L}^w}{\partial \boldsymbol{\Psi}_c} \frac{d \boldsymbol{\Psi}_c}{d\zeta} \\
\frac{d\mathbf{L}^b}{d\zeta} & = & \frac{\partial \mathbf{L}^b}{\partial \boldsymbol{U}^b} \frac{d \boldsymbol{U}^b}{d\zeta} 
 + \frac{\partial \mathbf{L}^b}{\partial \boldsymbol{\Psi}_c} \frac{d \boldsymbol{\Psi}_c}{d\zeta} 
\end{eqnarray}

The utility level in each community $\mathbf{U}^w,\mathbf{U}^b$ is driven by floor surface prices and by the social interaction parameter of near interactions.
\begin{eqnarray}
\frac{d\mathbf{U}^w}{d\zeta} & = & \frac{\partial \mathbf{U}^w}{\partial \boldsymbol{\Psi}_n} \frac{d \boldsymbol{\Psi}_n}{d\zeta} + \frac{\partial \mathbf{U}^w}{\partial \mathbf{q}_n} \frac{d \mathbf{q}_n}{d\zeta} \\
\frac{d\mathbf{U}^b}{d\zeta} & = & \frac{\partial \mathbf{U}^b}{\partial \boldsymbol{\Psi}_n} \frac{d \boldsymbol{\Psi}_n}{d\zeta} + \frac{\partial \mathbf{U}^b}{\partial \mathbf{q}_n} \frac{d \mathbf{q}_n}{d\zeta} 
\end{eqnarray}
The interactions across communities are such that the social equilibrium condition $\mathbf{f}_c$ needs to hold in every neighborhood:
\begin{eqnarray}
\frac{\partial \mathbf{f}_c}{\partial \boldsymbol{\Psi}_c} \frac{d\boldsymbol{\Psi}_c}{d\zeta} + \frac{\partial \mathbf{f}_c}{\partial \mathbf{U}^w} \frac{d\mathbf{U}^w}{d\zeta} & = 0 \label{equ:social_equilibrium_city}
\end{eqnarray}
The change in prices and the social interactions parameters in each community are found by requiring the social interaction conditions $\mathbf{f}_n$ and the market clearing conditions $\mathbf{m}_n$ to hold in each neighborhood:
\begin{eqnarray}
 0 &= &\frac{\partial \mathbf{f}_n}{\partial \boldsymbol{\Psi}_n} \frac{d\boldsymbol{\Psi}_n}{d\zeta} + \frac{\partial \mathbf{f}_n}{\partial \mathbf{q}_n} \frac{d\mathbf{q}_n}{d\zeta} \\
0 & = & \frac{\partial \mathbf{m}_n }{\partial \zeta} + \frac{\partial \mathbf{m}_n }{\partial \mathbf{q}_n} \frac{d \mathbf{q}_n}{d \zeta} + \frac{\partial \mathbf{m}_n }{\partial L^w} D \frac{d\mathbf{L}^w}{d\zeta} + \frac{\partial \mathbf{m}_n }{\partial L^b} D \frac{d\mathbf{L}^b}{d\zeta} + \frac{\partial \mathbf{m}_n}{\partial \boldsymbol{\Psi}_n} \frac{d \boldsymbol{\Psi}_n}{d \zeta} \label{equ:market_clearing_condition}
\end{eqnarray}
where $D$ is a design matrix of 0s and 1s of size $J_n \times J_c$ that maps population levels into the appropriate communities.

There is a price response for two reasons: first, a lower housing supply elasticity (a higher $\zeta$) leads to price increases $\partial \mathbf{q}_n / \partial \zeta$. This is the partial equilibrium shock that causes the city's population to shift across neighborhoods and communities. Second, there are readjustments in populations across communities and across neighborhoods within communities. This is the second term in $\frac{\partial \mathbf{q}_n }{\partial L}$.

The fixed point equilibrium condition at the community level is $\mathbf{f}_c$ in the system above. This corresponds to:
\begin{equation}
    \mathbf{f}_c(\boldsymbol{\Psi}_c, \boldsymbol{U}^w) = \left[ \Psi^{i}_c - \sum_{\iota =1}^n e^{-\xi d_{i\iota}}  \frac{\Psi_\iota^{\gamma\theta} (U^{w\iota})^\theta}{\sum_{k=1}^N \Psi_k^{\gamma\theta} (U^{wk})^\theta } \right]_i \label{eq:fixed_point_community_level}
\end{equation}
where $d_{i\iota}$ is the distance between two communities $i$ and $\iota$ within the same region. The fixed point equilibrium condition within each community is $\mathbf{f}_n$. This corresponds to, for each community $i$:
\begin{equation}
\mathbf{f}_n^i(\boldsymbol{\Psi}_n^i, \mathbf{q}_i) = \left[ \Psi^{ij}_n - \sum_{j'=1}^n e^{-\xi d_{ijj'}}  \frac{(\Psi^{ij'}_n)^{\gamma\theta} (q_{ij'})^{-\alpha\theta} A_{ij'}^\theta}{\sum_{k=1}^N (\Psi^{ik}_n)^{\gamma\theta} (q_{ik})^{-\alpha\theta} A_{ik}^\theta } \right]_j
\end{equation}
where $d_{ijj'}$ is the distance between two neighborhoods within the same community $i$.

Finally, we prove the proposition. Starting with the market clearing condition (\ref{equ:market_clearing_condition}),
\begin{eqnarray}
    \frac{d \mathbf{q}_n}{d\zeta} & = & -\left( \frac{\partial \mathbf{m}_n}{ \partial \mathbf{q}_n} \right)^{-1} \left[ \frac{\partial \mathbf{m}_n}{\partial \zeta} + \frac{\partial \mathbf{m}_n}{\partial \mathbf{L}^w} D \frac{d\mathbf{L}^w}{d\zeta} + \frac{\partial \mathbf{m}_n}{\partial \mathbf{L}^b} D \frac{d\mathbf{L}^b}{d\zeta} + \frac{\partial \mathbf{m}_n}{\partial \boldsymbol{\Psi}_n} \frac{d\boldsymbol{\Psi}_n}{d\zeta}\right] \label{equ:differential_qn_unsimplified}
\end{eqnarray}
and we notice three facts. First, the social equilibrium condition implies:
\begin{equation}
    \frac{d\boldsymbol{\Psi}_n}{d \zeta} = -\left( \frac{\partial \mathbf{f}_n}{\partial \boldsymbol{\Psi}_n} \right)^{-1} \frac{\partial \mathbf{f}_n}{\partial \mathbf{q}_n} \frac{d \mathbf{q}_n}{d \zeta}. \label{equ:social_equilibrium_wrapped}
\end{equation}
Second, the differential of utility implies the existence of matrices:
\begin{equation}
    \frac{d \mathbf{L}^w}{d\zeta} = E^w \frac{d \mathbf{U}^w}{d \zeta}, \qquad \frac{d \mathbf{U}^w}{d\zeta} = F^w \frac{d \mathbf{q}_n}{d\zeta}  \label{equ:matrices}
\end{equation}
and similarly for $b$. With $E^w = \frac{\partial L^w}{\partial U^w} - \frac{\partial L^w}{\partial \Psi_c} \left(  \frac{\partial f_c}{\partial \Psi_c} \right)^{-1} \frac{\partial f_c}{\partial U^w}$, and $F^w = -\frac{\partial U^w}{\partial \Psi_n} \left(\frac{\partial f_n}{\partial \Psi_n} \right)^{-1} \frac{\partial f_n}{\partial q_n} + \frac{\partial U^w}{\partial q_n}$.

Together, equations \ref{equ:differential_qn_unsimplified}, \ref{equ:social_equilibrium_wrapped}, and \ref{equ:matrices} allow an expression of all terms of the rhs of \ref{equ:differential_qn_unsimplified} as either constants or functions of $\frac{d \mathbf{q}_n}{d\zeta}$. We get:
\begin{equation}
    \left(\textrm{Id}-\Phi \right) \frac{d \mathbf{q}_n}{d\zeta} = -\left( \frac{\partial \mathbf{m}_n}{\partial \mathbf{q}_n} \right)^{-1} \frac{\partial \mathbf{m}_n}{ \partial \zeta}
\end{equation}
with $\Phi = - \left( \frac{\partial \mathbf{m}_n}{\partial \mathbf{q}_n} \right)^{-1} \left[ \frac{\partial \mathbf{m}_n}{\partial \mathbf{L}^w} DE^w F^w + \frac{\partial \mathbf{m}_n}{\partial \mathbf{L}^b} DE^b F^b - \frac{\partial \mathbf{m}_n}{\partial \boldsymbol{\Psi}_n} \left( \frac{\partial f_n}{\partial \Psi_n} \right)^{-1} \frac{\partial f_n}{\partial q_n} \right]$. This concludes the first part of the proof.

Use equation (\ref{equ:matrices}) to integrate and find population levels: $\mathbf{L}^w(\zeta) = \mathbf{L}^w(0) + \int_0^\zeta E^wF^w \frac{d \mathbf{q}_n}{d \zeta} ds$. Similar relationships hold for the vectors $\mathbf{L}^b$, $\mathbf{U}^w$ (equation \ref{equ:matrices}), $\mathbf{U}^b$, $\boldsymbol{\Psi}_c$ (equation (\ref{equ:social_equilibrium_city})), and $\boldsymbol{\Psi}_n$ (equation (\ref{equ:social_equilibrium_wrapped})). This concludes the second part of the proof.
\end{proof}

The differential equation (\ref{eq:citywide_ode}) is then solved. The initial condition at $\zeta=0$ ($\eta=\infty$) is one of the equilibrium vectors found earlier for a perfectly elastic floor supply. The differential equation is solved by a Runge-Kutta approach with predictor-corrector method as previously in the paper.

\subsection{Empirical Results}

\paragraph*{Data}

The analysis is performed on 353 neighborhoods grouped into 77 communities and 9 regions of the City of Chicago. This is presented on Figure~\ref{fig:chicago_map}, where the dashed lines are the boundaries of the neighborhoods, the solid lines are the boundaries of the communities, and the colors are for the regions. Such communities include locations that were prominently featured in the economics and sociology literatures, such as Humboldt Park, Logan Circle, Irving Park, Pilsen, Uptown. The regions are delineated by the City of Chicago.

We build a panel of the 353 neighborhoods followed from 1940 to 2010 with consistent boundaries. This is done by building a Census tract relationship file for 1940 to 2010, 1950 to 2010, 1960 to 2010, and then using the Geolytics panel for 1970 to 2010. We group tracts into 353 neighborhoods, so that within each community, households have approximately 5 locations to choose from. The variables include the total population, the black population, the price of housing, the decade. 

\paragraph*{Estimation of the Parameters} Social interactions operate within a certain scope $\xi$ as in \citeasnoun{redding2017quantitative}, denoted $\delta$ in \citeasnoun{ahlfeldt2015economics} and called the rate of decay. And therefore an important task is to calculate:
\begin{equation}
    \Psi_{j} = \sum_{k=1}^J e^{-\xi d_{jk}} L_{k}^w \label{eq:definition_psi_chicago}
\end{equation}
where $d_{jk}$ is the distance between neighborhood $j$ and neighborhood $k$. We use the centroid in the US National Atlas projection. $\xi$ is estimated as follows. For a fixed value of $\xi$, estimate the regression of log population by race:
\begin{eqnarray}
    \log L_{jt}^{\textrm{race}} & = & 1(\textrm{race} = \textrm{black}) \gamma^b \log \Psi_{jt-1}  + 1(\textrm{race} = \textrm{white})  \gamma^w \log \Psi_{jt-1} \nonumber \\
    && + 1(\textrm{race} = \textrm{black}) \nonumber \\ 
    && + \textrm{Neighborhood}_{j} + \textrm{Year}_t + \textrm{Residual}_{j}^{\textrm{race}} 
\end{eqnarray}
with fixed effects for the location, the year, and double clustering residuals by neighborhood and year. Then find the $\xi$ that maximizes the fit of this regression. This is presented on Figure~\ref{fig:scope_estimation}, suggesting that the demographic composition of neighborhoods beyond the immediate neighborhood matter as well.

\paragraph*{Multiple Equilibria} We then turn to the estimation of the equilibria within each community. The amenity vector is obtained by noticing that the neighborhood fixed effect is the $\log A_j$ in this paper's notation. $\gamma^w$ and $\gamma^b$ are the estimates obtained from the regression and each $\Psi_j$ is computed using (\ref{eq:definition_psi_chicago}) and the estimate of $\xi$.

As described, we start by enumerating the equilibria when floor surface supply is perfectly elastic, and then gradually reduce supply elasticity $\eta$. Figure~\ref{fig:equilibria_one_community} presents the five equilibria thus obtained in the case of the community of Humboldt Park. The shades of blue correspond to the share of the White population that lives in each location. We can see that because of the strength of social interactions relative to the amenities, the distribution of white population is not uniquely pinned down by the vector of amenities.

Figure~\ref{fig:equilibria_of_region} turns to the equilibria of the West Side region, which includes Humboldt Park. The Figure presents four equilibria. For each community of the region (Austin, West Town, West Garfield Park, East Garfield Park, Near West Side, North Lawndale, Lower West Side, South Lawndale), we select an equilibrium. The maps suggest that the variance of amenity fixed effects is not large enough to pin down the equilibrium. There exist equilibria where, for instance, North Lawndale attracts white residents.

Equilibrium multiplicity is not guaranteed. Table~\ref{tab:number_equilibria_community} presents the number of equilibria obtained with each of the 77 communities. The number of equilibria vary from 1 (single equilibrium), in the case of Archer Heights, Armour Square, Auburn Gresham, Avalon Park, and 37 other communities, to 9, in the case of Austin. 36 communities out 77 exhibit multiple equilibria. Parameters are identified using the historical shifts in population and is thus informed by the history of gentrification and white flight in these communities.

\paragraph*{Citywide Equilibria}  We then turn to finding \emph{citywide} equilibria. As noted before, the problem can be split into, first, the problem of finding equilibria across the 353 neighborhoods within each of the 77 communities, then the equilibria across communities within each of the 9 regions. In a second step, using a citywide equilibrium, we lower floor surface supply elasticity smoothly along a path using the homotopy (a multivariate ordinary differential equation). The end point is an equilibrium vector, a citywide equilibrium when supply elasticity has a selected value.

For each community $i$, we compute the welfare $V_i(1,\eta=\infty)$ for a perfectly elastic housing supply, and using the ordinary differential equation, compute the welfare $V_i(1,\eta)$ for lower supply elasticities $\eta<\infty$. We then calculate the citywide population distribution across communities $L_i(\eta)$ for each floor surface elasticity. 

The results are as follows. Figure \ref{fig:equilibria_of_region} depicts four potential counterfactual spatial equilibria across the 45 neighborhoods of the West Side region, with perfectly elastic floor surface supply. Figure \ref{fig:equilibrium_path_lower_supply_elasticity} maps the evolution of the equilibrium of the city when $\eta$ goes from $\infty$ to $1$. The starting point of such homotopy is the equilibrium of Figure~\ref{fig:equilibrium_3_west_side}.  Prices are a source of strategic substitutability (as in a Cournot model) and therefore lead to spatial equilibria with less spatial segregation of white residents.  This is apparent when comparing the initial equilibrium of Figure~\ref{fig:psi_n_initial_equilibrium} and the impact of a lower supply elasticity on social demographics $\Psi_n$ on Figure~\ref{fig:impact_eta_psi_n}: there is a lower spatial concentration in locations where strategic complementarities had led to a high share of white households. 

The methods presented can be extended to heterogeneous supply elasticities across locations using a parameterization $\eta_i(t)$ for $t\in[0,1]$ and perform the homotopy along $t$ instead of along $1/\eta$.

\begin{figure}
    \caption{Equilibria of the West Side Region when $\eta = \infty$}
    \label{fig:equilibria_of_region}

        \emph{Four potential equilibria for the West Side region (see map of regions of the City of Chicago on Figure~\ref{fig:chicago_map}) when floor surface supply is perfectly elastic. In each community, an equilibrium was selected. For instance, see the equilibrium of Humboldt Park on Figure~\ref{fig:equilibria_one_community}. The equilibria of the West side region were obtained by solving (\ref{eq:fixed_point}).}
    
    \begin{center}    
    \begin{subfigure}{0.45\textwidth}
    \caption{Equilibrium 1}
    \includegraphics[scale=0.5]{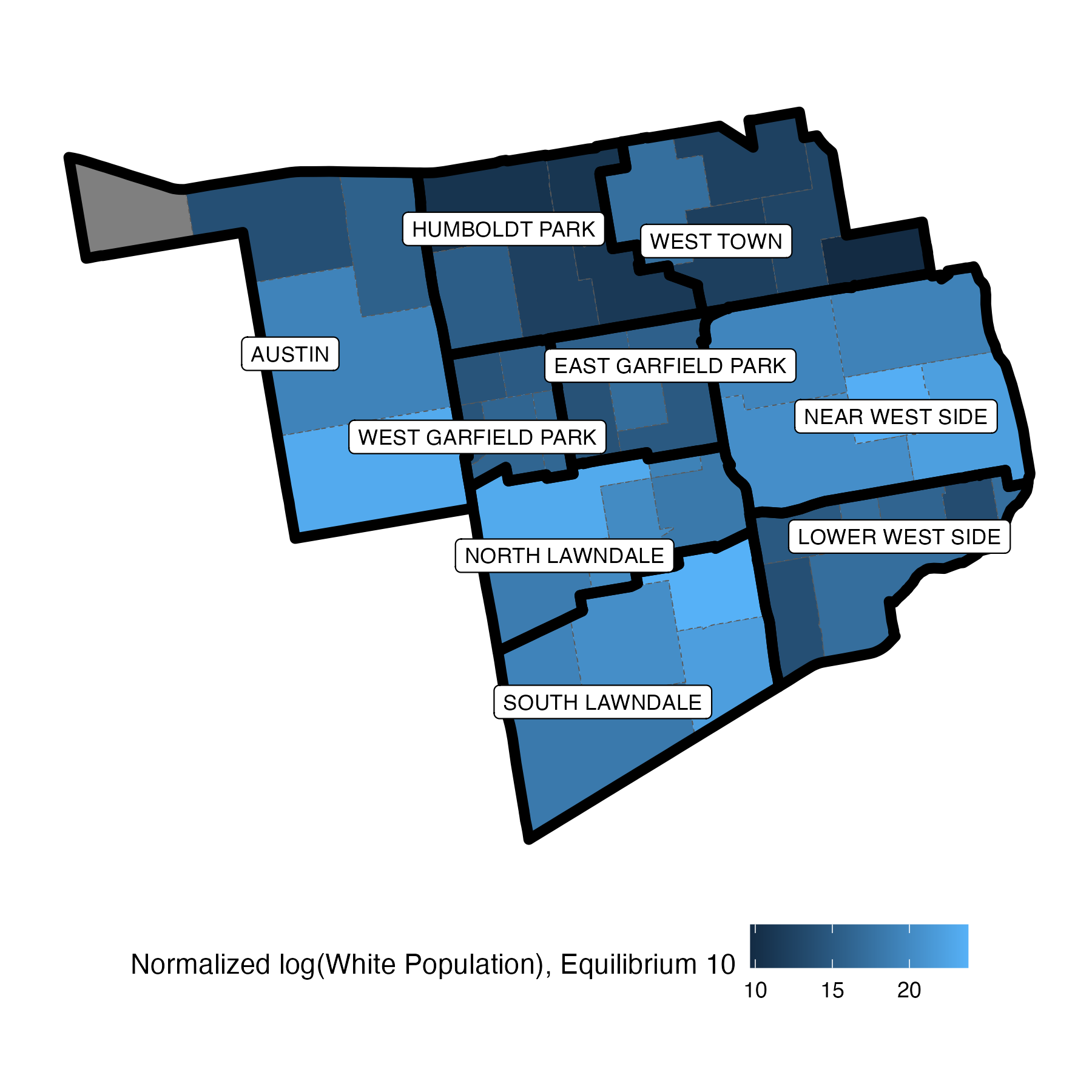}\qquad
    \end{subfigure}
    \begin{subfigure}{0.45\textwidth}
    \caption{Equilibrium 2}\label{fig:equilibrium_2_west_side}
    \includegraphics[scale=0.5]{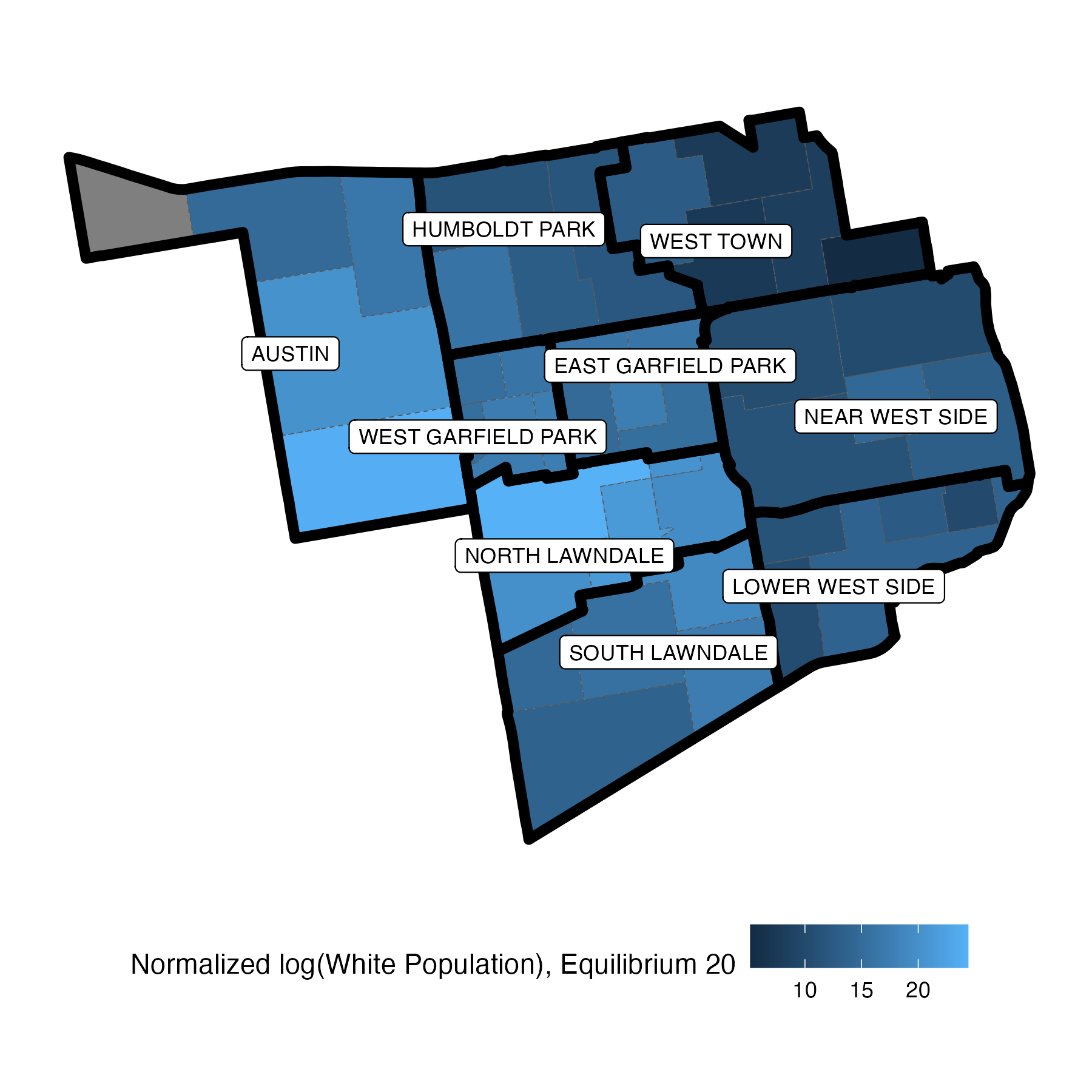}
    \end{subfigure}

    \begin{subfigure}{0.45\textwidth}
    \caption{Equilibrium 3} \label{fig:equilibrium_3_west_side}
    \includegraphics[scale=0.5]{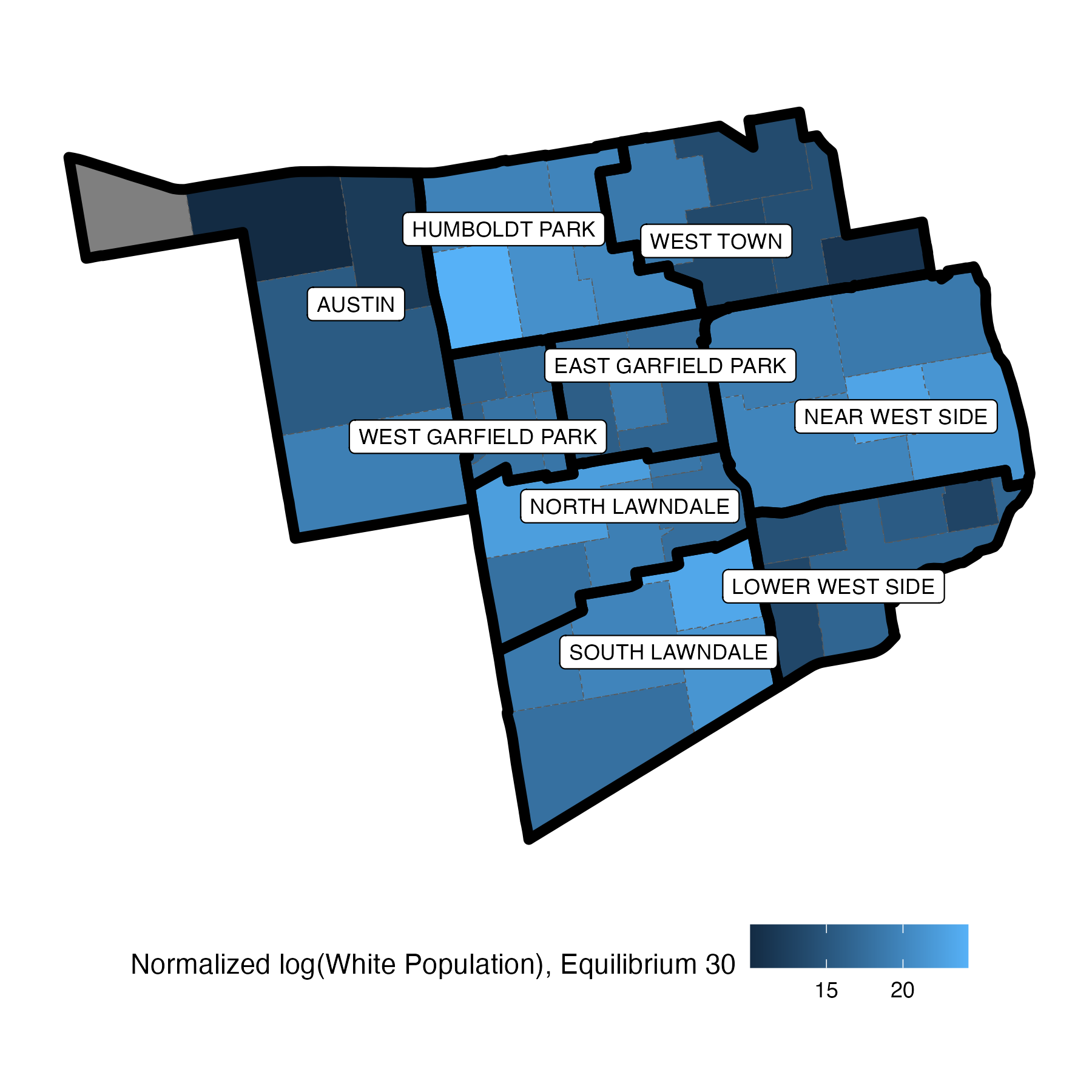}\qquad
    \end{subfigure}
    \begin{subfigure}{0.45\textwidth}
    \caption{Equilibrium 4}
    \includegraphics[scale=0.5]{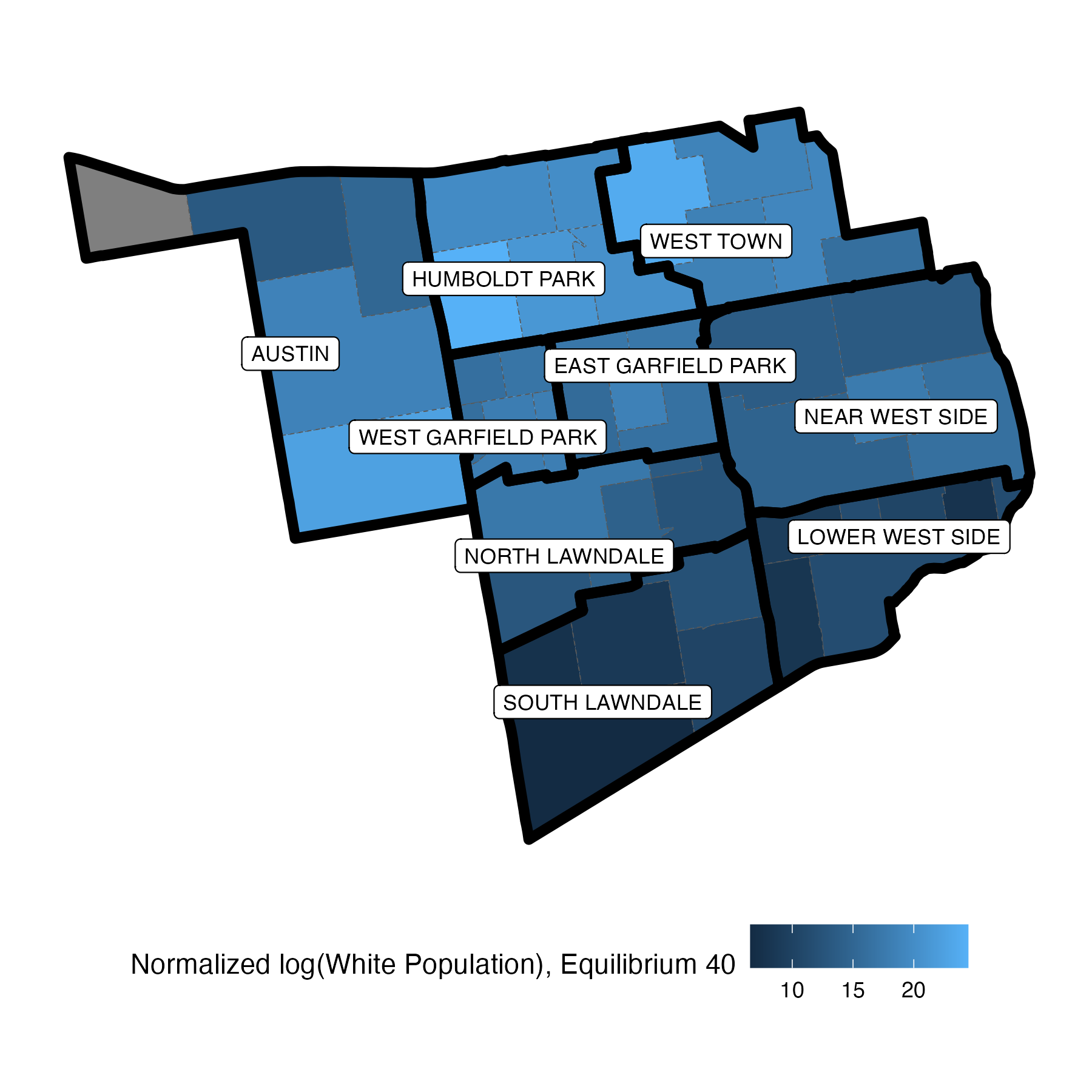}
    \end{subfigure}
    \end{center}
    
\end{figure}


\begin{figure}
    \caption{Equilibrium Path when Lowering Supply Elasticity}
    \label{fig:equilibrium_path_lower_supply_elasticity}

    \emph{These figures describe the evolution of the urban equilibrium from one selected spatial equilibrium 
    (Figure~\ref{fig:equilibrium_3_west_side}), as the paths of equilibria as supply elasticity decreases from its initial 
    $\eta=\infty$. Figure~\ref{fig:path_inverse_elasticity} displays the evolution of inverse elasticity from $\zeta = 1/\eta = 0$ to $\zeta = 1$ along the path. Figure~\ref{fig:smallest_eigenvalue_jacobian} displays the smallest eigenvalue of $(1-\Phi)$. Invertibility of $1-\Phi$ ensures the existence of a unique path. Figure (\ref{fig:psi_n_initial_equilibrium}) presents the spatial distribution of the the weighted neighbors' demographics $\Psi_n$ at the initial equilibrium ($\eta=\infty$) and Figure (\ref{fig:impact_eta_psi_n}) presents the impact of a lower supply elasticity (and higher prices) on such demographics $\Psi_n$ at the final equilibrium.}

    \centering
    \begin{subfigure}{0.45\textwidth}
    \caption{Inverse Elasticity along the Path} \label{fig:path_inverse_elasticity}
    \includegraphics[scale=0.5]{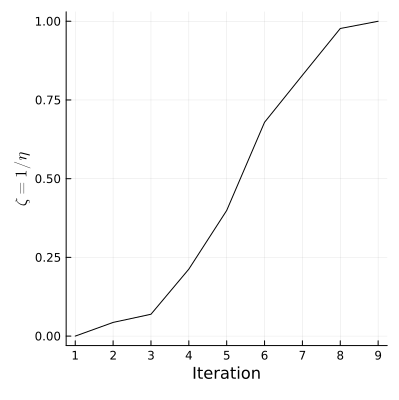}
    \end{subfigure}
    \qquad
    \begin{subfigure}{0.45\textwidth}
    \caption{Smallest Eigenvalue of the Jacobian}\label{fig:smallest_eigenvalue_jacobian}\label{fig:Diagnostic-of-Potential}
    \includegraphics[scale=0.5]{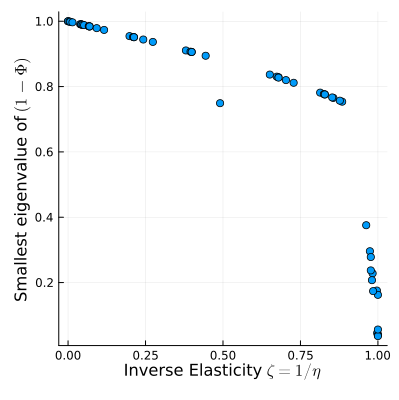}
    \end{subfigure}

    \bigskip
    
    \begin{subfigure}{0.45\textwidth}
    \caption{Initial $\Psi_n$}\label{fig:psi_n_initial_equilibrium}
    \includegraphics[scale=0.7]{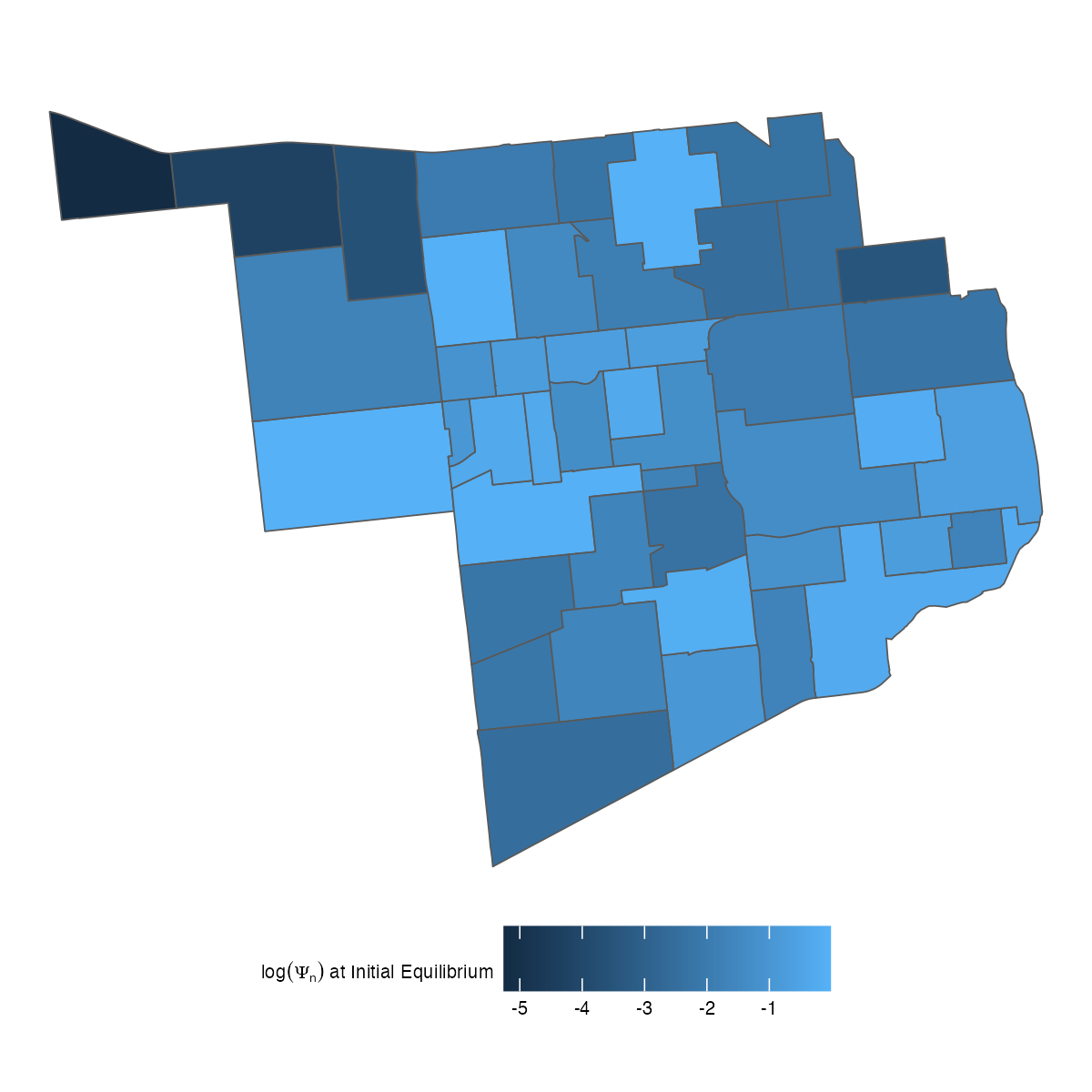}
    \end{subfigure}
    \qquad
    \begin{subfigure}{0.45\textwidth}
    \caption{Impact of Lower Elasticity $\eta$ on $\Psi_n$}\label{fig:impact_eta_psi_n}
    \includegraphics[scale=0.7]{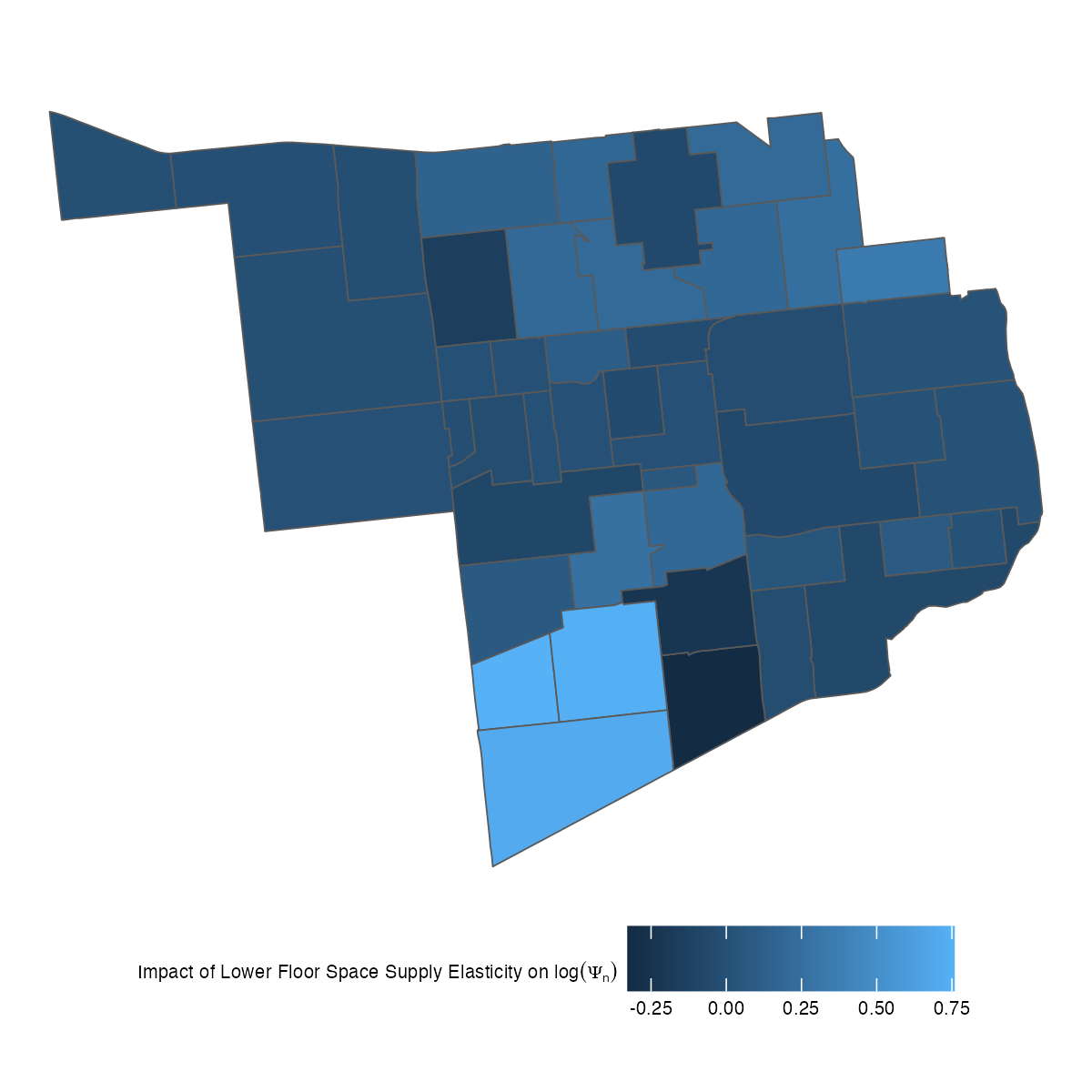}
    \end{subfigure}

\end{figure}

\section{Discussion and Extensions}\label{sec:discussion}

\subsection{Revisiting Brock and Durlauf (2001): Many Choices}\label{subsec:mcfadden}

Often utility functions are not Cobb-Douglas. This will be the case
in the non-homothetic examples recently developed in seminal literature. In other cases, the workhorse McFadden approach is used.

Consider a value function of functional form similar to \citeasnoun{brock2001discrete} and \citeasnoun{mcfadden1973conditional}: 
\begin{equation}
V_{j}^{g}=A_{j}-\alpha q_{j}+\gamma^{g}\Psi_{j}+\varepsilon_{ij}\label{eq:mcfadden_utility}
\end{equation}
with, as before, $\Psi_{j}=\sum_{k=1}^{J}e^{-\xi d_{jk}}x_{k}$ and
$x_{j}=N^{1}P(V_{j}^{1}>\max\{V_{k}^{1},k\neq j\})/s_{j}$ the density
of group $1$ at location $j$. $\varepsilon_{ij}$ is iid extreme-value
distributed. This is a generalized version of Brock and Durlauf, which
considered a binary choice (e.g. smoking or not smoking) and a local
scope of interactions. 

The social equilibrium equations can be written as: 
\begin{equation}
\Psi_{j}=\sum_{k=1}^{J}e^{-d_{jk}}x_{k}=\sum_{k=1}^{J}e^{-\xi d_{jk}}\frac{N^{1}}{s_{k}}\frac{\exp(A_{k}-\alpha q_{k}+\gamma^{1}\Psi_{k})}{\sum_{\ell}\exp(A_{\ell}-\alpha q_{\ell}+\gamma^{1}\Psi_{\ell})}\label{eq:social_conditions}
\end{equation}
 The exponential $\exp(y)$ has a well-known Taylor expansion, also
known as MacLaurin series:
\begin{equation}
\exp(y)=\sum_{p=0}^{\infty}\frac{y^{p}}{p!}\label{eq:taylor_expansion_exponential}
\end{equation}
Yet we will apply this Taylor expansion to the city $\mathcal{C}_{h}^{\infty}$
with perfectly elastic housing supply. The social equilibrium conditions become:
\begin{equation}
\sum_{\ell}\exp(A_\ell + \gamma^{1}\Psi_{\ell})\Psi_{j}=\sum_{k=1}^{J}e^{-\xi d_{jk}} \frac{N^1}{s_k} \exp(A_k + \gamma^{1}\Psi_{k})\label{eq:social_conditions_homogeneous_city}
\end{equation}
or, when using MacLaurin series:
\begin{equation}
\sum_{\ell=1}^J\sum_{p=0}^{\infty}\frac{(A_\ell + \gamma^{1}\Psi_{\ell})^{p}}{p!}\Psi_{j}=\sum_{k=1}^{J}e^{-\xi d_{jk}}\frac{N^1}{s_k} \sum_{p=0}^{\infty}\frac{(A_k + \gamma^{1}\Psi_{k})^{p}}{p!}\label{eq:equilibrium_mac_laurin_series}
\end{equation}
This shows that the degree of the polynomial system is the degree of the Taylor expansion plus
1. The comprehensive set of equilibria of this system can be obtained by total degree
homotopy.

The order $p$ of the polynomial expansion is finite and thus a key question is whether the homotopy at a given $p$ converges to the homotopy for an infinite number of polynomial term expansions. 

\begin{prop}\textbf{(Uniform Convergence of the Homotopy Function Using the MacLaurin--Taylor Expansion)}
Consider the total degree homotopy $H^n$, built using the $n$-th order MacLaurin series~(\ref{eq:equilibrium_mac_laurin_series}). 
\begin{equation}
H_j^n(\boldsymbol{\Psi},t) = t \Psi_j^{n+1} + (1-t) \left\{ \sum_{\ell=1}^J\sum_{p=0}^{n}\frac{(A_\ell + \gamma^{1}\Psi_{\ell})^{p}}{p!}\Psi_{j} - \sum_{k=1}^{J}e^{-\xi d_{jk}}\frac{N^1}{s_k} \sum_{p=0}^{n}\frac{(A_k + \gamma^{1}\Psi_{k})^{p}}{p!} \right\} \label{eq:equilibrium_mac_laurin_series_2} 
\end{equation}
This is a sequence $H^n(\boldsymbol{\Psi},t)$ of functions from $(0,\infty)^{2J}\times[0,1]$ to $\mathbb{R}$ indexed by $n\in\mathbb{N}$. The vector $\boldsymbol{\Psi}$ belongs to a compact set as it is positive and bounded upwards by $\Delta N^1$. This sequence $H^n$ thus converges uniformly to the homotopy $H(\boldsymbol{\Psi},t)$ using McFadden's logit function. For each $\varepsilon>0$, there exists a $\underbar{n}$ such that $n>\underline{n}$ implies that: $H^n(\boldsymbol{\Psi},t)$ is within $\varepsilon$ of $H(\boldsymbol{\Psi},t)$ for all $\boldsymbol{\Psi}$ in its compact set.
\end{prop}

The key lemma is that the Taylor series of the exponential is uniformly convergent to the exponential on the compact set, by Weierstrass' M-test. A good reference is \citeasnoun{gamelincomplex}.\footnote{Chapter XIII.}



The final step, i.e. the homotopy from perfectly elastic $\eta=\infty$
to elastic $\eta\geq0$ obtains the equilibria of the city $\mathcal{C}$
in the \citeasnoun{mcfadden1973conditional} and \citeasnoun{brock2001discrete} setup. 

Table~\ref{tab:mcfadden_results} presents the equilibria obtained when applying the method using the accompanying code. We use a city with 4 locations, homogeneous amenities, and an 8th degree expansion of the exponential. The upper panel is for the perfectly elastic city, and reports the demographics of each of the four neighborhoods (columns) for each of the 15 equilibria (rows). In each case, the population condition is satisfied, i.e. the total population of type 1 matches the aggregate population of type~1. 

As amenities are homogeous, they do not pin down the equilibrium and this is visible on the table: equilibria are ranked in lexicographic order of $x_1,x_2,x_3,x_4$. We see that equilibrium 1 and equilibrium 15 predict a large degree of stratification across locations, but in equilibrium 1 it is location 4 that hosts the largest number of type 1 households -- while in equilibrium 15 it is location 1 that hosts the largest number of type 1 households. Equilibria 7 and 8 predict a more equal distribution of type 1 households across locations. This suggests that, with the parameters thus supplied, the city with only 4 locations has a large number of equilibria.

\subsection{Dynamics}

Static social interaction models are typically steady-state approximations of a dynamic model \cite{blume2003equilibrium}. This paper's approach in terms of polynomial systems and homotopies can be extended to the
case of a dynamic city.  We describe here a potential approach in the case where there is
one demographic group and households have a preference for densely populated
locations. 
\begin{equation}
V_{jt}=A_{j}q_{jt}^{-\alpha}\Psi_{jt}^{\gamma}\varepsilon_{ij}\Pi_{jt+1}^{\delta}\label{eq:valuation_dynamic}
\end{equation}
We set the Fr\'echet dispersion parameter $\theta=1$ without loss of
generality to lighten the notations. $\Psi_{j}$ is a weighted average
of the density levels in neighboring locations:
\begin{equation}
\Psi_{jt}=\sum_{k=1}^{J}e^{-\xi d_{jk}}\frac{L_{kt}}{s_{k}}\label{eq:social_parameter}
\end{equation}
where $d_{jk}$ is the distance of location $j$ to location $k$,
$A_{k}$ is the surface of location $k$ and $L_{k}$ is the population
choosing location $j$. Households face a mobility cost $\mu_{jk}$
when moving from location $j$ to location $k$. Thus the welfare
in $j$ depends as follows on the utilities in each location:
\begin{equation}
\Pi_{jt}=\sum_{k=1}^{J}A_{k}q_{kt}^{-\alpha}\Psi_{kt}^{\gamma}\Pi_{kt+1}^{\delta}/\mu_{jk}\label{eq:bellman_equations}
\end{equation}
At equilibrium, the population level $L_{kt}$ satisfies:
\begin{equation}
L_{jt+1}=\sum_{k=1}^{J}L_{kt}\frac{A_{j}q_{jt}^{-\alpha}\Psi_{jt}^{\gamma}\Pi_{jt+1}^{\delta}/\mu_{kj}}{\sum_{l=1}^{J}A_{l}q_{lt}^{-\alpha}\Psi_{lt}^{\gamma}\Pi_{lt+1}^{\delta}/\mu_{kl}}\label{eq:population_flows}
\end{equation}
A typical approach to these types of dynamic models is to first start
by finding the steady-states. Here we provide an approach to obtain
the multiple steady-state equilibria consistent with the exogenous
parameters $(\mathbf{A},\alpha,\gamma,\delta,\xi,\boldsymbol{\mu})$. 

The social equilibrium and the Bellman equations can be expressed
as a polynomial system. First, find a rational number $p^{\delta}/q^{\delta}$,
where $p^{\delta}\in\mathbb{N}$ and $q^{\delta}\in\mathbb{N}^{*}$
that approximates $\delta$, and $p^{\delta},q^{\delta}$ have no
common divisor. This is always possible as $\mathbb{Q}$ is dense
in $\mathbb{R}$. Then, perform the change of variables $\Xi_{j}=\Pi_{j}^{1/q^{\delta}}$.
Perform the same approach for $\gamma$ and $\Psi_{jt}$: find a rational
approximation $p/q$ of $\gamma$ and perform the change of variable
$z_{jt}=\Psi_{jt}^{1/q}$. 

Express the Bellman equation (\ref{eq:bellman_equations}) and the population flow (\ref{eq:population_flows}) as polynomial
systems of the coefficients in amenities and population in the city
with perfectly elastic supply. The polynomial system in $(\mathbf{z},\Xi)\in\mathbb{R}^{J}\times\mathbb{R}^{J}$
has $2J$ equations. Finally, perform a differentiable homotopy from
the city with elastic housing supply $\eta=\infty$ to the city with
finite elasticity $\eta\in(0,\infty)$. When the Jacobian is invertible along the path, 
this provides a path $(\mathbf{z}(t),\Xi(t),\mathbf{q}(t))\in\mathbb{R}^{J}\times\mathbb{R}^{J}\times\mathbb{R}^{J}$
whose endpoint $(\mathbf{z}(1),\Xi(1),\mathbf{q}(1))$ is an equilibrium
of the dynamic city with elastic housing supply as long as the Jacobian
is invertible along the path.

\subsection{Bifurcations at Singular Points of the Jacobian}\label{sec:singularity_jacobian}\label{sec:bifurcations}

The approach described in this Section relies on first-order differential equations for the equilibrium. These differential equations have a unique solution at each point as long as the Jacobian is of full rank. In equation~\ref{eq:differential_equation}, the matrix $\Gamma$ may not be invertible and thus such first-order differential equation may not have a unique solution. In a limited number of cases in our simulations of the homotopy when transitioning from the perfectly elastic to the finite elasticity city, we observe that the smallest eigenvalue of the Jacobian approaches zero. At this point the differential equation may have multiple solutions, and some additional ``discipline'' may be required to compute the multiple paths that emerges at this singular point. The approach described here harkens back to \citeasnoun{keller1977numerical}. 

The approach relies on two facts: the first order condition defines the subspace of paths that are consistent with such first order condition; when the Jacobian is not invertible, this subspace has dimension at least 1. The second fact is that the set of potential solutions can be narrowed by considering the differential equation implied by a second-order Taylor expansion. Finding the set of solutions to this second order differential equation that lie in the kernel of the first-order equation provides a typically finite number of bifurcation paths at the singularity point. We implement the approach in open source Julia code. The second-order differential equation is a second order polynomial in the differentials $d\mathbf{z}$, $d\mathbf{q}$ of social demographics and prices for each location. These conditions can thus be solved using homotopy approaches for polynomial equations. 

\subsubsection*{Finding the Singular Point}
Numerically the path may approach a point where the smallest eigenvalue of the Jacobian is less than a threshold $\varepsilon$. At this point, a Newton-Raphson approach can be used to solve:
\begin{equation}
    \textrm{det} \left( \frac{\partial \mathbf{H}}{\partial \mathbf{z}}(\mathbf{z}, t) \right) = 0
\end{equation}
This can be solved using a Newton-Raphson approach.
\begin{equation}
    t_{n+1} = t_n - \left\{ \frac{d}{dt} \textrm{det} \left( \frac{\partial \mathbf{H}}{\partial \mathbf{z}}(\mathbf{z}, t) \right) \right\}^{-1} \textrm{det} \left( \frac{\partial \mathbf{H}}{\partial \mathbf{z}}(\mathbf{z}, t) \right)
\end{equation}
To perform this, notice that the differential of the determinant can be expressed as:
\begin{equation}
    \frac{d}{dt} \textrm{det} \left( \frac{\partial \mathbf{H}}{\partial \mathbf{z}}(\mathbf{z}, t) \right) = \textrm{det} \left[ \frac{\partial \mathbf{H}}{\partial \mathbf{z}}(\mathbf{z}, t) \right] \textrm{Tr}\left[ \left(\frac{\partial \mathbf{H}}{\partial \mathbf{z}}(\mathbf{z}, t)\right)^{-1} \underbrace{\frac{d}{dt} \frac{\partial \mathbf{H}}{\partial \mathbf{z}}}_{\textrm{See below}} \right]
\end{equation}
with\begin{equation}
    \frac{d}{dt} \frac{\partial \mathbf{H}}{\partial \mathbf{z}} = \frac{\partial^2 \mathbf{H}}{\partial \mathbf{z}^2} \frac{d\mathbf{z}}{dt} + \frac{\partial^2 \mathbf{H}}{\partial \mathbf{z} \partial t}
\end{equation}
The next step is to solve for the bifurcation points.

\begin{figure}
    \caption{Bifurcation method}
    \label{fig:enter-label}
\emph{The singular point is the solution to $det(\partial H/\partial \mathbf{z}) = 0$. The subspace for which $\partial H/\partial \mathbf{z}$ is singular is the space $\text{Ker} \partial H/\partial \mathbf{z}$ depicted below. The three solutions to the system of equations implied by the second-order Taylor expansion are depicted as $d\mathbf{z}_1$, $d\mathbf{z}_2$, $d\mathbf{z}_3$. Here only $d\mathbf{z}_1$, $d\mathbf{z}_2$ belong to the space implied by the first-order condition, and thus are solutions to the differential equation at the singular point.}

\centering

\includegraphics[scale=0.5]{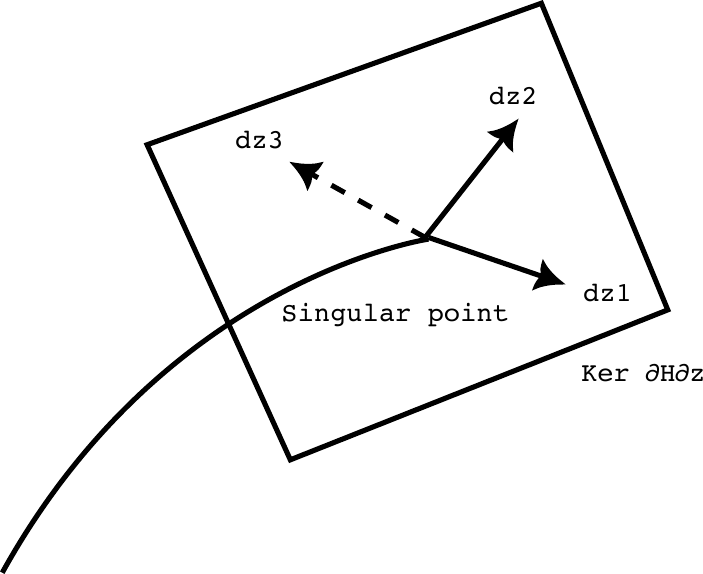}

\end{figure}

\subsubsection*{Finding the Bifurcations at the Singular Point}

At the singular point, the Jacobian is non invertible and its kernel is therefore of dimension $\geq 1$. There is at least a potential line (dim 1) of $d\mathbf{z}$ consistent with the first-order differential equation. We thus need to narrow down to a smaller set of potential $d\mathbf{z}$. First denote by $\textrm{Ker}(\frac{\partial H}{\partial \mathbf z})$ the kernel of the Jacobian. The bifurcations $d\mathbf{z}$ lie in this kernel.

We need a second set of constraints. At the singular point $\tilde{\mathbf{z}}$ we perform a second-order Taylor expansion of the homotopy $H$.
\begin{equation}
    \textbf{H}(\tilde{\mathbf{z}}+d\mathbf{z}, t+dt) - \textbf{H}(\tilde{\mathbf{z}}, t) = 0
\end{equation}
becomes, for element $j$ of $H$:
\begin{equation}
    \frac{\partial H_j}{\partial \mathbf{z}} d\mathbf{z} + \frac{\partial H_j}{\partial t} dt + \frac{1}{2} d\mathbf{z}' \frac{\partial^2 H_j}{\partial \mathbf{z}^2} d\mathbf{z} + d\mathbf{z}' \frac{\partial^2 H_j}{\partial \mathbf{z} \partial t} dt + \frac{\partial^2 H_j}{\partial t^2} (dt)^2 = 0
\end{equation}
which is a system of $j$ equations in the unknown vector $d\mathbf{z} \in \mathbb{R}^J$ knowing the scalar $dt \in \mathbb{R}$.

This can be written as a polynomial system of order 2. For each $j=1,2,\dots,J$:
\begin{equation}
    \sum_k \frac{\partial H_j}{\partial z_k} dz_k + \frac{\partial H_j}{\partial t} dt + \frac{1}{2} \sum_{k,l}  \frac{\partial^2 H_j}{\partial z_k \partial z_l} dz_k dz_l + \sum_k  \frac{\partial^2 H_j}{\partial z_k \partial t} dt dz_k + \frac{\partial^2 H_j}{\partial t^2} (dt)^2 = 0
\end{equation}
which can be reorganized to highlight the monomials in increasing power:
{\footnotesize
\begin{equation}
    \left[ \frac{\partial H_j}{\partial t} dt + \frac{\partial^2 H_j}{\partial t^2} (dt)^2 \right] + \sum_k \left[\frac{\partial^2 H_j}{\partial z_k \partial t} dt + \frac{\partial H_j}{\partial z_k} \right] dz_k + \frac{1}{2} \sum_{k,l} \left[ \frac{\partial^2 H_j}{\partial z_k \partial z_l} \right] dz_k dz_l   = 0
\end{equation}}
Bezout's theorem tells us that there can be up to $2^K$ solutions. We operate here in $\mathbb{R}^K$ and therefore consider the real solutions of the system that lie in the kernel of the Jacobian. These are the bifurcations at the singular point $\tilde{\mathbf{z}}$. A homotopy can thus be used to solve for those bifurcation points.

\section{Conclusion}

This paper provides an approach to equilibrium multiplicity in discrete choice models, in the context of quantitative spatial models. Results suggest the existence of counterfactual cities consistent with the set of exogenous amenities. 

The approach is based on the tools of algebraic geometry and topology by acknowledging that economies with social interactions, heterogeneity, and price responses are ``similar'' (homotopic) to simpler economies for which we can enumerate all equilibria. The methods have lower dimensionality than typical total degree homotopy used in other scientific fields such as the methods of \citeasnoun{sommese2005numerical}. 

The methods of this paper should be useful to the broader research community. The implication of linear algebra have led to tremendous progress, including the simulation of impulse response functions in dynamic economies in macroeconomics \cite{dejong2012structural} and in urban economics and trade \cite{kleinman2023linear}. The tools of algebraic geometry may also lead to substantial progress; they can be seen as an extension of linear algebra to polynomials, and in particular one that finds the zeros of polynomials. This should be useful in a broad range of applications beyond the ones described in this paper. More advanced results of algebraic geometry can provide more insights. For instance, \possessivecite{Hilbert1893} \emph{Nullstellensatz} could potentially characterize the amenities of a city solely by its equilibria. The result, applied to discrete choice economies, establishes a one-to-one relationship between the zeros of the polynomial system defining the equilibria and its vector of amenities, social preference parameters, and the scope of social interactions. More advanced results in algebraic geometry should be useful for the identification and estimation of economies with multiple equilibria.

\bibliographystyle{agsm}
\bibliography{multiple_equilibria}
\clearpage\pagebreak{}





\begin{table}
\caption{Finding City Equilibria by Homotopy -- Social Interactions and Perfectly
Elastic Supply, City $\mathcal{C}^\infty$ \label{tab:Finding-Equilibria-by-Method-1}}

\bigskip

\emph{This table provides the outcome of the calculation of the equilibria
with the exact method presented in Section~\ref{subsec:Perfectly-Elastic-Housing-Supply-Exact-Method}.
This method starts with a polynomial whose roots are on the complex
unit circle. For each such root, the solution is updated along the
path $t\in[0,1]$ towards the equilibria of the city. Bertini's theorem
guarantees that this approach provides all equilibria. $J$ is the
number of locations, $p/q=\gamma^{g}$ is the preferences of college
educated workers for college-educated neighbors, $\xi$ is the scope
of social interactions, as $\Delta_{jk} = e^{-\xi d_{jk}}$ is the weight of $k$
in the parameter $\Psi_{j}=\sum e^{-\xi d_{jk}}x_{k}$, and $\sigma(A)$
is the standard deviation of the log normal amenities $\log A_{j}\sim N(0,\sigma(A))$.
We keep $\alpha=0.3$, $\eta=\infty$, $L_g=0.8J$, $L - L_g=0.2J$, $mc=1$,
and $d_{jk}=|j-k|$.\bigskip{}
}

\begin{center}
\begin{tabular}{ccccc}
\toprule
   $J$  &  $\gamma^1$ &  $\xi$ & $\sigma(\log(A))$ & $N^{eq}$ \\
   Locations & Social Pref. & Scope of Interactions & S.D. of Amenities & \# of Proper Equilibria \\
\midrule
  3  &  0.2& 3.0&    0.5&    1\\
  3  &  0.5& 0.1&    0.5&    1\\
  3  &  2.0& 4.0&    0.1&    7\\
  3  &  2.0& 4.0&    0.5&    7\\
  3  &  2.0& 1.0&    0.5&    1\\
  3  &  4.0& 4.0&    0.5&    7\\
  3  &  5.0& 4.0&    1.0&    7\\
  3  &  5.0& 4.0&    0.5&    7\\
  5  & 1.25& 0.1&    1.0&    1\\
  5  &  2.0& 3.0&    0.1&   31\\
  5  &  2.0& 4.0&    1.5&   15\\
  7  &  1.0& 1.0&    1.0&    1\\
  7  &  2.5& 1.0&    0.0&   13\\
  7  &  2.5& 1.0&    0.0&   13\\
  7  &  4.0& 2.0&    0.5&  127 \\
\bottomrule
\end{tabular}
\end{center}

\emph{\bigskip{}
}
\end{table}

\clearpage\pagebreak{}





\appendix








\begin{figure}

\caption{A Geometric Representation of the Equilibria in the Two-Location Case: Singular Points, Equilibrium Multiplicity\label{fig:two-location-case-geometry}}

\bigskip
\emph{These three figures present the zeros of the polynomial system for the equilibrium of the city of Section~. The \textcolor{orange}{orange curve} is for the equilibrium equation for location 1. The \textcolor{blue}{blue curve} is the equilibrium equation for location 2. Both are polynomial equations in $z_1,z_2$. These figures illustrate two points: (a)~as the strength of social preferences $\gamma^1$ increases (from subfigure (a) to subfigure (b)), new equilibria appear. (b)~as amenities become heterogeneous (here $A_2$ the amenity of location 2 increases from subfigure (b) to subfigure(c)), equilibria disappear. At the value of $A_2$ for which the two curves have a tangency point, the Jacobian $\partial H / \partial \mathbf{z}$ of the homotopy $H$ w.r.t $z_1,z_2$ becomes rank deficient, i.e. has a zero determinant and a zero eigenvalue.}
\bigskip 

\centering
\begin{subfigure}[b]{0.48\textwidth}
\centering
\caption{\tiny Single Equilibrium, Homogeneous Locations} \bigskip 
\includegraphics[scale=0.5]{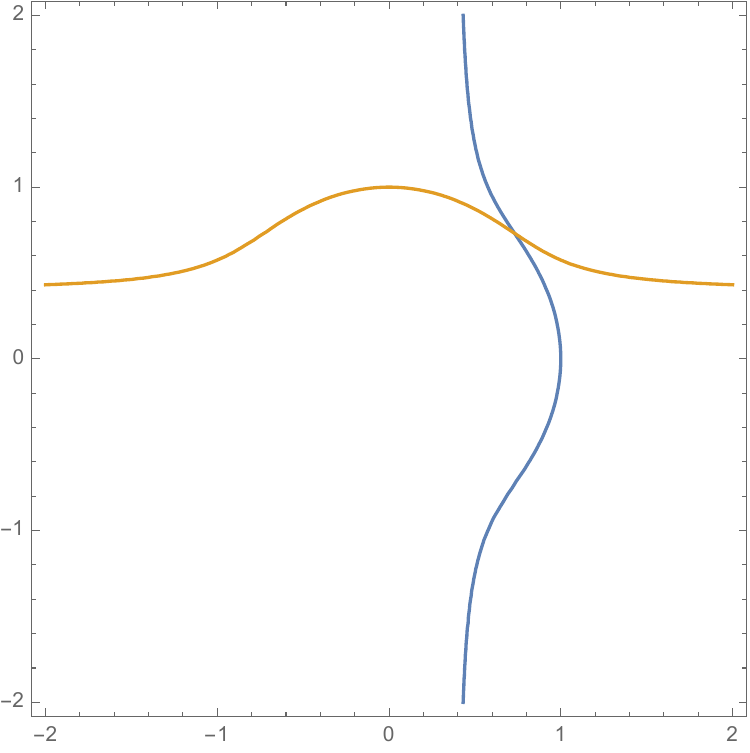}
\end{subfigure}

\bigskip
\bigskip

\begin{subfigure}[b]{0.48\textwidth}
\centering 
\caption{\tiny Homogeneous Locations: Three Equilibria} \bigskip 
\includegraphics[scale=0.5]{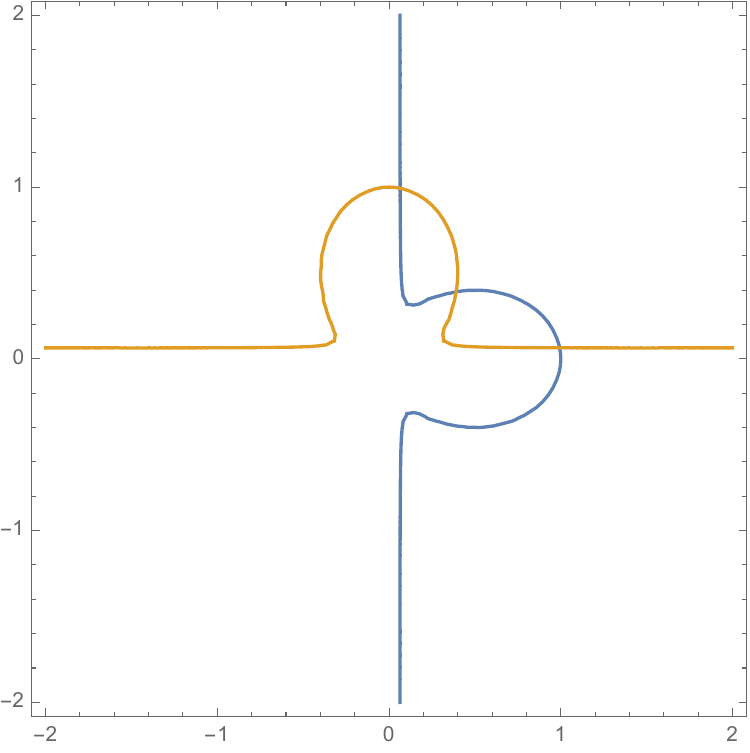}
\end{subfigure}
\hfill
\begin{subfigure}[b]{0.48\textwidth}
\centering 
\caption{\tiny Heterogeneous Locations: A Singular Equilibrium where the Jacobian is not invertible} \bigskip 
\includegraphics[scale=0.5]{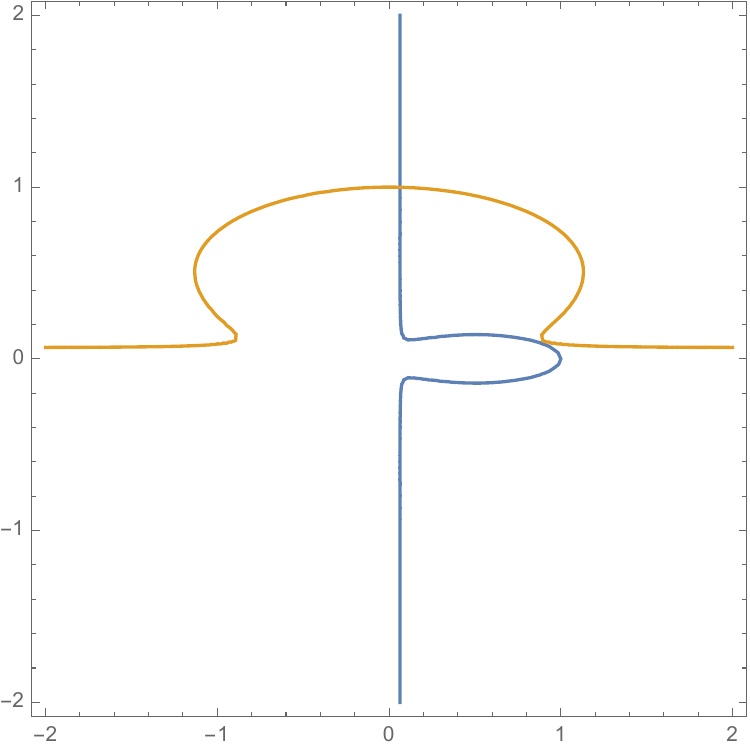}
\end{subfigure}

\end{figure}

\clearpage\pagebreak{}

\begin{figure}

\caption{A Geometric Representation of the Equilibria in the Three-Location Case\label{fig:three-location-case-geometry}}

\bigskip

\emph{In the three location case, the intuitions developed for the two-location case (Figure~\ref{fig:two-location-case-geometry}) carries through. First, as we increase the social preference parameter $\gamma^1$, the number of equilibria increases from 1 to 5. The shape of the polynomial surface defined by each of the three locations' (the affine variety of each location) is similar to the two-location case. When such surfaces are tangent, the Jacobian $\partial H/\partial \mathbf{z}$ of the homotopy $H$ w.r.t. $z_1$, $z_2$, $z_3$ is rank deficient.}

\bigskip

\centering

\begin{subfigure}[b]{0.48\textwidth}
\caption{Lower Preference $\gamma^1$, Single Equilibrium}
\includegraphics[scale=0.5]{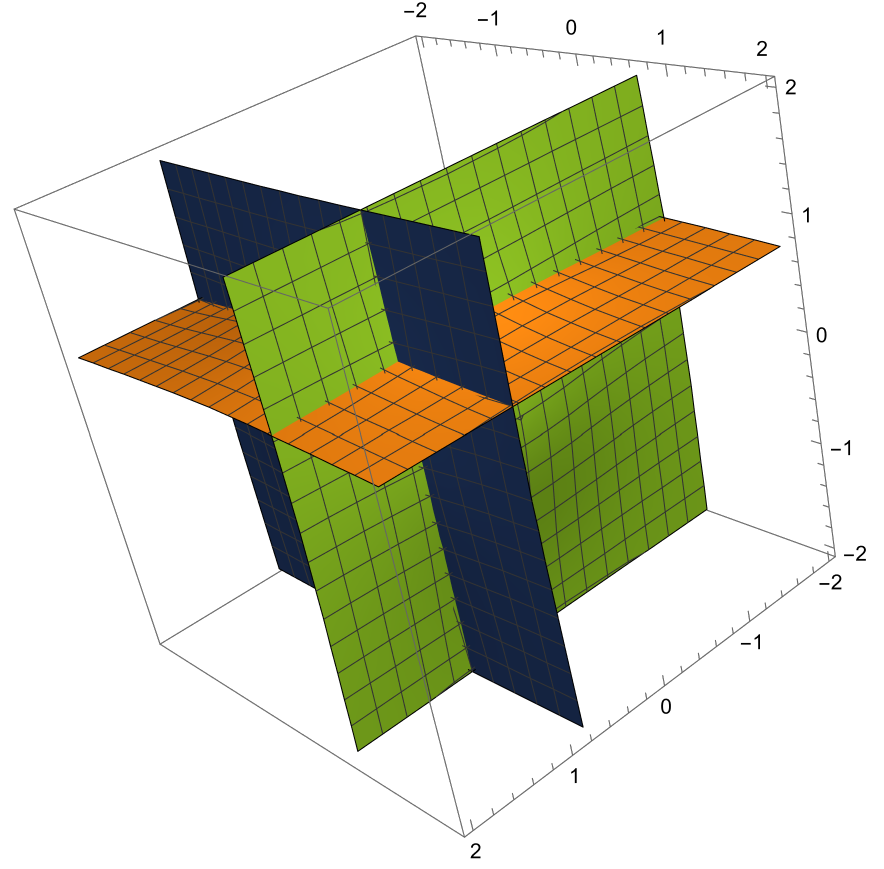}
\end{subfigure}

\bigskip
\bigskip

\begin{subfigure}[b]{0.48\textwidth}
\caption{Higher Preference $\gamma^1$, Multiple Equilibria}    
\includegraphics[scale=0.35]{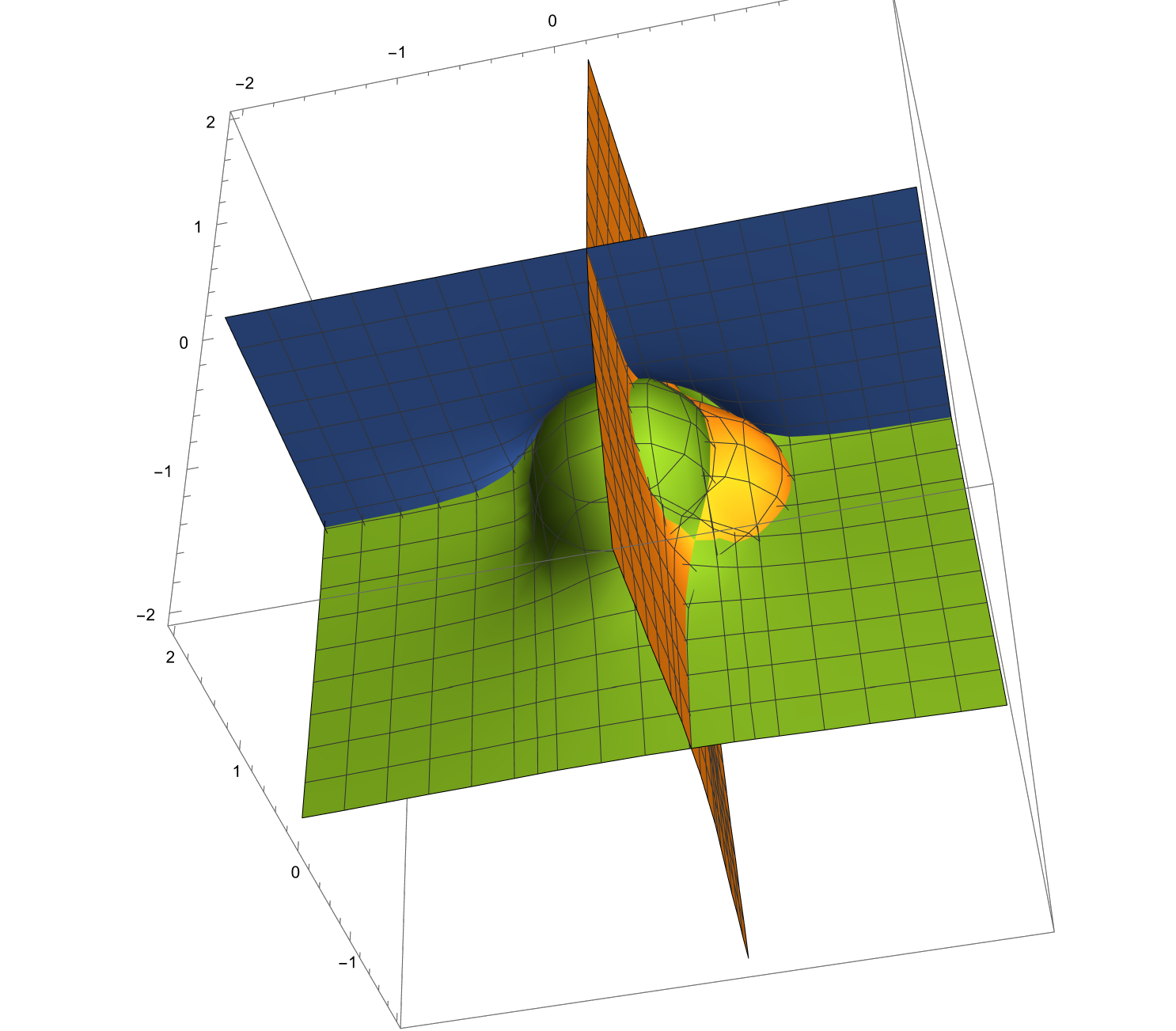}
\end{subfigure}

\end{figure}

\clearpage\pagebreak{}

\begin{table}

\caption{The McFadden - Brock and Durlauf Case of Section~\ref{subsec:mcfadden}, The Example of 15 Equilibria for One City}\label{tab:mcfadden_results}

\bigskip

\emph{This table provides the outcome of an example of calculation of equilibria in the case of McFadden and Brock and Durlauf's logit social equilibrium. The table is for the case of a city with $J=4$ locations, social preferences $\gamma^1=5$, amenities $A_j=1$, distance $\Delta_{ij} = \exp(-2 \left| i-j \right|)$ and perfectly elastic housing supply $\eta=\infty$. This is solved using total degree homotopy and an 8th order expansion of the exponential. The equilibria are solved in $\boldsymbol{\Psi}$, and then $\mathbf{x} = \Delta^{-1} \boldsymbol{\Psi}$. The equilibria are in lexicographic order of $\mathbf{x}$.}

\bigskip

\begin{subfigure}[b]{0.9\textwidth}

\begin{center}

\begin{tabular}{ccccc}
\toprule
& \multicolumn{4}{c}{Equilibrium Social Demographics}\\
\cmidrule(lr){2-5}
 Equilibrium \#  &                  $x_1$ &                  $x_2$ &                  $x_3$ &                  $x_4$ \\  
 \midrule
          1 & 0.008&  0.008& 0.009&   0.974\\
          2 &   0.008& 0.009&   0.973& 0.009\\
          3 & 0.009&   0.973& 0.009& 0.008\\
          4 &  0.068&   0.072&   0.429&   0.431\\
          5 &  0.072&  0.428&   0.428&  0.072\\
          6 &  0.078&   0.417&  0.084&  0.420\\
          7 &  0.126&   0.292&   0.282&   0.300\\
          8 &   0.262&  0.238&  0.238&  0.262\\
          9 &   0.286&   0.275&  0.146&  0.293\\
         10 &   0.293&  0.146&  0.275&   0.286\\
         11 &   0.300&   0.282&  0.293&  0.126\\
         12 &  0.420&  0.084&   0.418&  0.078\\
         13 &   0.420&   0.080&   0.080&   0.420\\
         14 &   0.431&   0.429&  0.072&  0.068\\
         15 &   0.974& 0.009& 0.008& 0.008\\
\bottomrule
\end{tabular}
\end{center}

\end{subfigure}

\end{table}

\clearpage\pagebreak{}

\end{document}